\def\year{2020}\relax
\documentclass[letterpaper]{article} % DO NOT CHANGE THIS
\usepackage{aaai20}  % DO NOT CHANGE THIS
\usepackage{times}  % DO NOT CHANGE THIS
\usepackage{helvet} % DO NOT CHANGE THIS
\usepackage{courier}  % DO NOT CHANGE THIS
\usepackage[hyphens]{url}  % DO NOT CHANGE THIS
\usepackage{graphicx} % DO NOT CHANGE THIS
\usepackage{graphicx}  % DO NOT CHANGE THIS
\usepackage{amsmath,amssymb,color,amsmath, graphicx, float, subcaption, mathrsfs,color}
\usepackage{amsthm}

\theoremstyle{plain}
\newtheorem{theorem}{\textbf{Theorem}}

\newtheorem{proposition}[theorem]{\textbf{Proposition}}

\newcommand{\citet}[1]{\citeauthor{#1} \shortcite{#1}}

\def\A{\mathbf{A}}

\def\G{\mathbf{G}}

\def\I{\mathbf{I}}

\def\P{\mathbf{P}}
\def\Q{\mathbf{Q}}

\def\W{\mathbf{W}}

\def\e{\mathbf{e}}

\def\x{\mathbf{x}}
\def\y{\mathbf{y}}
\def\z{\mathbf{z}}

\def\btheta{\boldsymbol{\theta}}

\def\argmin#1{\underset{#1}{\textrm{argmin}}}

\def\minim#1{\underset{#1}{\textrm{min}}}

\title{Image-Adaptive GAN based Reconstruction}
\author{Shady Abu Hussein\thanks{The authors have contributed equally to this work. Code is available at https://github.com/shadyabh/IAGAN}, Tom Tirer$^*$, and Raja Giryes\\ \Large \\ % All authors must be in the same font size and format. Use \Large and \textbf to achieve this result when breaking a line
School of Electrical Engineering\\ %If you have multiple authors and multiple affiliations
Tel Aviv University, Tel Aviv, Israel\\
}

\begin{document}

\maketitle

\begin{abstract}
In the recent years, there has been a significant improvement in the quality of samples produced by (deep) generative models such as variational auto-encoders and generative adversarial networks. However, the representation capabilities of these methods still do not capture the full distribution for complex classes of images, such as human faces.
This deficiency has been clearly observed in previous works that use pre-trained generative models to solve imaging inverse problems.
In this paper, we suggest to mitigate the limited representation capabilities of generators by making them image-adaptive and enforcing compliance of the restoration with the observations via back-projections.
We empirically demonstrate the advantages of our proposed approach for image super-resolution and compressed sensing.
\end{abstract}

\section{Introduction}

The developments in deep learning \cite{goodfellow2016deep} in the recent years have led to significant improvement in learning generative models.
Methods like variational auto-encoders (VAEs) \cite{kingma2013auto}, generative adversarial networks (GANs) \cite{goodfellow2014generative} and latent space optimizations (GLOs) \cite{bojanowski2018optimizing} have found success at modeling data distributions. However, for complex classes of images, such as human faces, while these methods can generate nice examples, their representation capabilities do not capture the full distribution. This phenomenon is sometimes referred to in the literature, especially in the context of GANs, as {\em mode collapse} \cite{arjovsky2017wasserstein,karras2017progressive}. Yet, as demonstrated in \cite{richardson2018gans}, it is common to other recent learning approaches as well.

Another line of works that has gained a lot from the developments in deep learning is imaging inverse problems, where the goal is to recover an image $\x$ from its degraded or compressed observations $\y$ \cite{bertero1998introduction}.
Most of these works have been focused on training a convolutional neural network (CNN) to learn the inverse mapping from $\y$ to $\x$ for a {\em specific} observation model (e.g. super-resolution with certain scale factor and bicubic anti-aliasing kernel \cite{dong2014learning}).
Yet, recent works have suggested to use neural networks for handling only the image prior in a way that does not require exhaustive offline training for each different observation model. This can be done by using CNN denoisers \cite{zhang2017learning,meinhardt2017learning,rick2017one} plugged into iterative optimization schemes \cite{venkatakrishnan2013plug,metzler2016denoising,tirer2019image}, training a neural network from scratch for the imaging task directly on the test image (based on internal recurrence of information inside a single image) \cite{ZSSR,ulyanov2017deep}, or using generative models \cite{bora2017compressed,yeh2017semantic,hand2018phase}.

Methods that use generative models as priors can only handle images that belong to the class or classes on which the model was trained. However, the generative learning equips them with valuable semantic information that other strategies lack. For example, a method which is not based on a generative model cannot produce a perceptually pleasing image of human face if the eyes are completely missing in an inpainting task \cite{yeh2017semantic}.
The main drawback in restoring complex images
using generative models is the limited representation capabilities of the generators. Even when one searches over the range of a pre-trained generator for an image which is closest to the original $\x$, he is expected to get a significant mismatch \cite{bora2017compressed}.

In this work, we propose a strategy to mitigate the limited representation capabilities of generators when solving inverse problems. The strategy is based on a gentle internal learning phase at test time, which essentially makes the generator image-adaptive while maintaining the useful information obtained in the offline training. In addition, in scenarios with low noise level, we propose to further improve the reconstruction by a back-projection step that strictly enforces compliance of the restoration with the observations $\y$.
We empirically demonstrate the advantages of our proposed approach for image super-resolution and compressed sensing.

\section{Related Work}

Our work is mostly related to the work by  \citet{bora2017compressed}, which have suggested to use pre-trained generative models for the compressive sensing (CS) task \cite{donoho2006compressed,candes2004robust}: reconstructing an unknown signal $\x \in \mathbb{R}^n$ from observations $\y \in \mathbb{R}^m$ of the form
\begin{align}
\label{Eq_model}
\y = \A \x + \e,
\end{align}
where $\A$ is an $m \times n$ measurement matrix, $\e \in \mathbb{R}^m$ represents the noise, and the number of measurements is much smaller than the ambient dimension of the signal, i.e. $m \ll n$.
Following the fact that in highly popular generative models (e.g. GANs, VAEs and GLOs) a generator $\G(\cdot)$ learns a mapping
from a low dimensional space $\z \in \mathbb{R}^k$ to the signal space $\G(\z) \subset \mathbb{R}^n$, the authors \cite{bora2017compressed} have proposed a method, termed CSGM, that estimates the signal as $\hat{\x}=\G(\hat{\z})$, where $\hat{\z}$ is obtained by minimizing the non-convex\footnote{The function $f(\z)$ is non-convex due to the non-convexity of $\G(\z)$.} cost function
\begin{align}
\label{Eq_cost_func1}
f(\z) = \| \y-\A \G(\z) \|_2^2,
\end{align}
using backpropagation and standard gradient based optimizers.

For specific classes of images, such as handwritten digits and human faces, the experiments in \cite{bora2017compressed,hand2018phase} have shown that using learned generative models enables to reconstruct nice looking images with much fewer measurements than methods that use non-generative (e.g. model-based) priors.
However, unlike the latter, it has been also shown that CSGM and its variants cannot provide accurate recovery even when there is no noise and the number of observations is very large. This shortcoming is mainly due to the limited representation capabilities of the generative models (see Section 6.3 in \cite{bora2017compressed}), and is common to related recent works \cite{hand2018phase,bora2018ambientgan,dhar2018modeling,shah2018solving}.

Note that using specific structures of $\A$, the model \eqref{Eq_model} can be used for different imaging inverse problems,
making the CSGM method applicable for these problems as well.
For example, it can be used for denoising task when $\A$ is the $n \times n$ identity matrix $\I_n$, inpainting task when $\A$ is an $m \times n$ sampling matrix (i.e. a selection of $m$ rows of $\I_n$), deblurring task when $\A$ is a blurring operator, and super-resolution task if $\A$ is a composite operator of blurring (i.e. anti-aliasing filtering) and down-sampling.

Our image-adaptive approach is inspired by \cite{tirer2018super}, which is influenced itself by \cite{ZSSR,ulyanov2017deep}. These works follow the idea of internal recurrence of information inside a single image within and across scales \cite{glasner2009super}.
However, while the two methods \cite{ZSSR,ulyanov2017deep} completely avoid an offline training phase and optimize the weights of a deep neural network only in the test phase, the work in \cite{tirer2018super} incorporates external and internal learning by taking offline trained CNN denoisers, fine-tuning them in test time and then plugging them into a model-based optimization scheme \cite{tirer2019image}.
Note, though, that the internal learning phase in \cite{tirer2018super} uses patches from $\y$ as the ground truth for a denoising loss function ($f(\tilde{\x}) = \| \y-\tilde{\x} \|_2^2$),
building on the assumption that $\y$ directly includes patterns which recur also in $\x$. Therefore, this method requires that $\y$ is not very degraded, which makes it suitable perhaps only for the super-resolution task, similarly to \cite{ZSSR}, which is also restricted to this problem.

Note that the method in \cite{ulyanov2017deep}, termed as deep image prior (DIP), can be applied to different observation models.
However, the advantage of our approach stems from the offline generative learning that captures valuable semantic information that DIP lacks. As mentioned above, a method like DIP, which is not based on a generative model, cannot produce a perceptually pleasing image of human face if the eyes are completely missing in an inpainting task \cite{yeh2017semantic}.
In this paper, we demonstrate that this advantage holds also for highly ill-posed scenarios in image super-resolution and compressed sensing.
In addition, note that the DIP approach typically works only with huge U-Nets like architectures that need to be modified for each inverse problem and require much more memory than common generators. Indeed, we struggled (GPU memory overflow, long run-time) to apply DIP to the $1024 \times 1024$ images of CelebA-HQ dataset \cite{karras2017progressive}.

\section{The Proposed Method}
\label{sec_proposed}

In this work, our goal is to make the solutions of inverse problems using generative models more faithful to the observations and more accurate, despite the limited representation capabilities of the pre-trained generators.
%To this end, we examine two strategies.
To this end, we propose an image-adaptive approach, whose motivation is explained both verbally and mathematically (building on theoretical results from \cite{bora2017compressed}). We also discuss a back-projection post-processing step that can further improve the results for scenarios with low noise level. While this post-processing, typically, only moderately improves the results of model-based super-resolution algorithms \cite{glasner2009super,yang2010image}, we will show that it is highly effective for generative priors.
To the best of our knowledge, we are the first to use it in reconstructions based on generative priors.

\subsection{An Image-Adaptive Approach}
%\subsection{"Soft" compliance to observations via image-adaptation}

%Our second strategy
We propose to handle the limited representation capabilities of the generators by making them image-adaptive (IA) using internal learning in test-time.
In details, instead of recovering the latent signal $\x$ as $\hat{\x}=\G(\hat{\z})$, where $\G(\cdot)$ is a pre-trained generator and $\hat{\z}$ is a minimizer of \eqref{Eq_cost_func1}, we suggest to simultaneously optimize $\z$ and the parameters of the generator, denoted as $\btheta$, by minimizing the cost function
\begin{align}
\label{Eq_cost_func_IA}
f_{IA}(\btheta,\z) = \| \y-\A \G_{\btheta}(\z) \|_2^2.
\end{align}
The optimization is done using backpropagation and standard gradient based optimizers.
The initial value of $\btheta$ is the pre-trained weights, and the initial value of $\z$ is $\hat{\z}$, obtained by minimization with respect to $\z$ alone, as done in CSGM.
Then, we perform joint-minimization to obtain $\hat{\btheta}_{IA}$ and $\hat{\z}_{IA}$, and recover the signal using $\hat{\x}_{IA}=\G_{\hat{\btheta}_{IA}}(\hat{\z}_{IA})$.

The rationale behind our approach can be explained as follows.
Current leading learning strategies cannot train a generator whose representation range covers {\em every} sample of a complex distribution, thus, optimizing only $\z$ is not enough.
However, the expressive power of deep neural networks (given by optimizing the weights $\btheta$ as well) allows to create a {\em single} specific sample that agrees with the observations $\y$.
Yet, contrary to prior works that optimize the weights of neural networks only by internal learning \cite{ZSSR,ulyanov2017deep}, here we incorporate information captured in the test-time with the valuable semantic knowledge obtained by the offline generative learning.

To make sure that the information captured in test-time does not come at the expense of offline information which is useful for the {\em test image at hand},
we start with optimizing $\z$ alone, as mentioned above,
and then apply the joint minimization
with a small learning rate and early stopping (details in the experiments section below).

\subsection{Mathematical Motivation for Image-Adaptation}

To motivate the image-adaptive approach, let us consider an $L$-layer neural network generator
\begin{align}
\label{Eq_generator}
\G(\z; \{\W_{\ell}\}_{\ell=1}^L) = \W_L \sigma ( \W_{L-1} \sigma ( \ldots \sigma ( \W_1 \z) \ldots ) ),
\end{align}
where $\sigma(\cdot)$ denotes element-wise ReLU activation, and $\W_\ell \in \mathbb{R}^{k_\ell \times k_{\ell-1}}$ such that
$k_L=n$.
Recall that typically $k_0 < k_1 < \ldots < k_L$ (as $k_0 \ll n$).
The following theorem, which has been proven in \cite{bora2017compressed} (Theorem 1.1 there), provides an upper bound on the reconstruction error.

\begin{theorem}
\label{theorem1}
Let $\G(\z):\mathbb{R}^k \rightarrow \mathbb{R}^n$ as given in \eqref{Eq_generator}, $\A \in \mathbb{R}^{m \times n}$ with $A_{ij}\sim \mathcal{N}(0,1/m)$, $m=\Omega \left ( kL\mathrm{log}n \right )$, and $\y=\A\x+\e$.
Let $\hat{\z}$ minimize $\| \y-\A \G(\z) \|_2$ to within additive $\epsilon$ of the optimum.
Then, with probability $1-\mathrm{e}^{-\Omega(m)}$ we have
\begin{equation}
\label{Eq_theorem1}
\| \G(\hat{\z}) - \x \|_2 \leq 6 E_{rep}(\G(\cdot),\x) + 3\|\e\|_2 + 2\epsilon,
\end{equation}
where $E_{rep}(\G(\cdot),\x) := \minim{\z}\|\G(\z)-\x\|_2$.
\end{theorem}
Note that $E_{rep}(\G(\cdot),\x)$ is in fact the representation error of the generator for the specific image $\x$.
This term has been empirically observed in \cite{bora2017compressed} to dominate the overall error, e.g. more than the error of the optimization algorithm (represented by $\epsilon$).
The following proposition builds on Theorem \ref{theorem1} and motivates the joint optimization of $\z$ and $\W_1$ by guaranteeing a decreased representation error term.

\begin{proposition}
\label{prop2}
Consider the generator defined in \eqref{Eq_generator} with $k_0 < k_1$, $\A \in \mathbb{R}^{m \times n}$ with $A_{ij}\sim \mathcal{N}(0,1/m)$, $m=\Omega \left ( k_1 L\mathrm{log}n \right )$, and $\y=\A\x+\e$. Let $\hat{\z}$ and $\hat{\W}_1$ minimize $\tilde{f}(\z,\W_1) = \| \y-\A \G(\z; \{\W_{\ell}\}_{\ell=1}^L) \|_2$ to within additive $\epsilon$ of the optimum.
Then, with probability $1-\mathrm{e}^{-\Omega(m)}$ we have
\begin{equation}
\label{Eq_prop2}
\| \G(\hat{\z}; \hat{\W}_1, \{\W_{\ell}\}_{\ell=2}^L) - \x \|_2 \leq 6 \tilde{E}_{rep} + 3\|\e\|_2 + 2\epsilon,
\end{equation}
where $\tilde{E}_{rep} \leq E_{rep}(\G(\cdot),\x)$.
\end{proposition}

\begin{proof}

Define $\hat{\tilde{\z}}:=\hat{\W}_1 \hat{\z}$ and $\tilde{\G}(\tilde{\z}) := \W_L \sigma ( \W_{L-1} \sigma ( \ldots \sigma ( \W_2 \sigma ( \I_{k_1} \tilde{\z}) ) \ldots ) )$.
Note that $\G(\hat{\z}; \hat{\W}_1, \{\W_{\ell}\}_{\ell=2}^L) = \tilde{\G}(\hat{\tilde{\z}})$, therefore $\hat{\tilde{\z}}$ minimize $\| \y-\A \tilde{\G}(\tilde{\z}) \|_2$ to within additive $\epsilon$ of the optimum.
Applying Theorem \ref{theorem1} on $\tilde{\G}(\tilde{\z})$ and $\hat{\tilde{\z}}$ we have
\begin{align}
\label{Eq_prop2_1}
\| \tilde{\G}(\hat{\tilde{\z}}) - \x \|_2 \leq 6 E_{rep}(\tilde{\G}(\cdot), \x) + 3\|\e\|_2 + 2\epsilon,
\end{align}
with the advertised probability.
Now, note that
\begin{align}
\label{Eq_prop2_2}
& E_{rep}(\tilde{\G}(\cdot),\x) \nonumber \\
& = \minim{\tilde{\z}\in\mathbb{R}^{k_1}}\|\W_L \sigma ( \W_{L-1} \sigma ( \ldots \sigma ( \W_2 \sigma (\I_{k_1} \tilde{\z})) \ldots ) ) - \x \|_2  \nonumber \\
&\leq \minim{\z\in\mathbb{R}^{k_0}}\|\W_L \sigma ( \W_{L-1} \sigma ( \ldots \sigma ( \W_2 \sigma (\W_1 \z)) \ldots ) ) - \x \|_2  \nonumber \\
&= E_{rep}(\G(\cdot),\x),
\end{align}
where the inequality follows from $\W_1\mathbb{R}^{k_0} \subset \mathbb{R}^{k_1}$.
We finish with substituting $\tilde{\G}(\hat{\tilde{\z}}) = \G(\hat{\z}; \hat{\W}_1, \{\W_{\ell}\}_{\ell=2}^L)$ in \eqref{Eq_prop2_1} and defining $\tilde{E}_{rep} := E_{rep}(\tilde{\G}(\cdot),\x)$.

\end{proof}

Proposition \ref{prop2} shows that under the mathematical framework of Theorem \ref{theorem1}, and under the (reasonable) assumption that the output dimension of the first layer is larger than its input, it is possible to {\em further} reduce the representation error of the generator for $\x$ (the term that empirically dominates the overall error) by optimizing the weights of the first layer as well. The result follows from obtaining an increased set in which the nearest neighbor of $\x$ is searched.

Note that if $k_0 < k_1 < \ldots < k_L$, then the procedure which is described in Proposition \ref{prop2} can be repeated sequentially layer after layer to further reduce the representation error.
However, note that this theory loses its meaningfulness at high layers because $m=\Omega \left ( k_\ell L\mathrm{log}n \right )$ approaches $\Omega \left ( n \right )$ (so no prior is necessary). Yet, it presents a motivation to optimize all the weights, as we suggest to do in practice.

\begin{figure*}[t]
  \centering
  \begin{subfigure}[b]{0.18\linewidth}
    \centering\includegraphics[width=66pt]{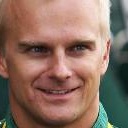}
    \\ \vspace{12mm}
  \end{subfigure}%
  \begin{subfigure}[b]{0.18\linewidth}
    \centering\includegraphics[width=66pt]{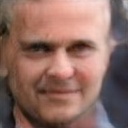}
    \\ \vspace{1mm}
    \centering\includegraphics[width=66pt]{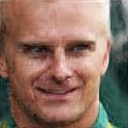}
  \end{subfigure}
    \begin{subfigure}[b]{0.18\linewidth}
    \centering\includegraphics[width=66pt]{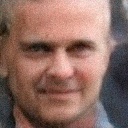}
    \\ \vspace{1mm}
    \centering\includegraphics[width=66pt]{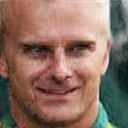}
  \end{subfigure}
    \begin{subfigure}[b]{0.18\linewidth}
    \centering\includegraphics[width=66pt]{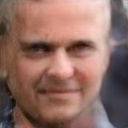}
    \\ \vspace{1mm}
    \centering\includegraphics[width=66pt]{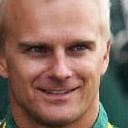}
  \end{subfigure}
    \begin{subfigure}[b]{0.18\linewidth}
    \centering\includegraphics[width=66pt]{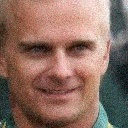}
    \\ \vspace{1mm}
    \centering\includegraphics[width=66pt]{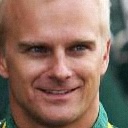}
  \end{subfigure}

  \caption{Compressed sensing with Gaussian measurement matrix using BEGAN. From left to right and top to bottom: original image, CSGM for $m/n=0.122$, CSGM-BP for $m/n=0.122$, CSGM for $m/n=0.61$, CSGM-BP for $m/n=0.61$, IAGAN for $m/n=0.122$, IAGAN-BP for $m/n=0.122$, IAGAN for $m/n=0.61$, IAGAN-BP for $m/n=0.61$.}
\label{fig:cs_mse_vs_m_visual}
\end{figure*}

\subsection{"Hard" vs. "Soft" Compliance to Observations}

%\subsection{"Hard" compliance to observations via back-projecting}
%\label{sec_proposed_BP}

The image-adaptive approach improves the agreement between the recovery and the observations. We turn now to describe another %(perhaps simpler)
complementary way to achieve this goal.

Denote by $\hat{\x}$ an estimation of $\x$, e.g. using CSGM method or our IA approach.
Assuming that there is no noise, i.e. $\e=0$, a simple post-processing to strictly enforce compliance of the restoration with the observations $\y$  %can be achieved by backward projecting
is back-projecting
(BP) the estimator $\hat{\x}$ onto the affine subspace $\{ \A \mathbb{R}^n = \y \}$
\begin{align}
\label{Eq_cost_func_csgm_proj}
\hat{\x}_{bp} &= \argmin{\tilde{\x}} \,\, \| \tilde{\x}-\hat{\x} \|_2^2  \,\,\,\, \textrm{s.t.} \,\,\,\, \A\tilde{\x}= \y. %  \nonumber  \\
%&= \A^\dagger (\y - \A \hat{\x}) + \hat{\x},
\end{align}
Note that this problem has a closed-form solution
\begin{align}
\label{Eq_csgm_proj}
\hat{\x}_{bp} &= \A^\dagger\y + (\I_n - \A^\dagger\A) \hat{\x} \nonumber  \\
&= \A^\dagger (\y - \A \hat{\x}) + \hat{\x},
\end{align}
where $\A^\dagger := \A^T(\A\A^T)^{-1}$ is the pseudoinverse of $\A$ (assuming that $m < n$, which is the common case, e.g. in super-resolution and compressed sensing tasks).
In practical cases, where the problem dimensions are high, the matrix inversion in $\A^\dagger$ can be avoided by using the conjugate gradient method %that only requires applying the operators $\A$ and $\A^T$
\cite{hestenes1952methods}.
%\textcolor{blue}{
Note that when $\y$ is noisy, the operation $\A^\dagger\y$ in \eqref{Eq_csgm_proj} is expected to amplify the noise. Therefore, the BP post-processing is useful as long as the noise level is low.
%}

\begin{figure}
    \centering\includegraphics[width=220pt]{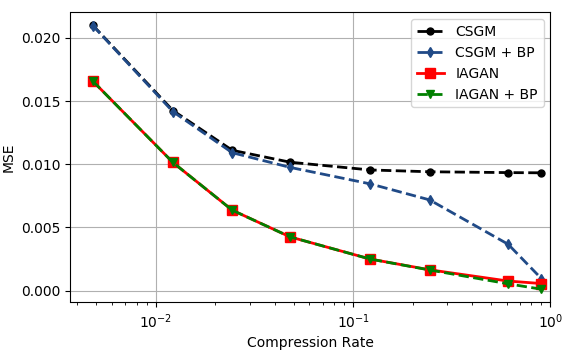}
  \caption{Compressed sensing with Gaussian measurement matrix using BEGAN. Reconstruction MSE (averaged over 100 images from CelebA) vs. the compression ratio $m/n$.}
\label{fig:cs_mse_vs_m}
\end{figure}

Let $\P_A := \A^\dagger \A$ denote the orthogonal projection onto the row space of $\A$, and $\Q_A := \I_n - \A^\dagger \A$ denote its orthogonal complement.
Substituting \eqref{Eq_model} into \eqref{Eq_csgm_proj} gives
\begin{align}
\label{Eq_csgm_proj_b}
\hat{\x}_{bp} = \P_A \x + \Q_A \hat{\x} + \A^\dagger \e,
\end{align}
which shows that $\hat{\x}_{bp}$ is consistent with $\y$ on $\P_A \x$ (i.e. displays {\em "hard" compliance}), and considers only the projection of $\hat{\x}$ onto the null space of $\A$.
Therefore, for an estimate $\hat{\x}$ obtained via a generative model,
the BP technique essentially {\em eliminates} the component of the generator's representation error that resides in the row space of $\A$, but does not change at all the component in the null space of $\A$.
Still, from the (Euclidean) accuracy point of view, this strategy is very effective at low noise levels, as demonstrated in the experiments section.

Interestingly, note that our image-adaptive strategy enforces only a {\em "soft" compliance} of the restoration with the observations $\y$, because our gentle joint optimization (which prevents overriding the offline semantic information) may not completely diminish the component of the generator's representation error that resides in the row space of $\A$,
as done by BP.
On the other hand, intuitively, the strong prior (imposed by the offline training and by the generator's structure) is expected to improve the restoration also in the null space of $\A$ (unlike BP).
Indeed, as shown below, by combining the two approaches, i.e. applying the IA phase and then the BP on $\hat{\x}_{IA}$, we obtain better results than only applying BP on CSGM. This obviously implies decreasing the component of reconstruction error in the null space of $\A$.

\section{Experiments}

In our experiments we use two recently proposed GAN models, which are known to generate very high quality samples of human faces.
The first is BEGAN \cite{berthelot2017began}, trained on CelebA dataset \cite{liu2015deep}, which generates a $128 \times 128$ image from a uniform random vector $\z \in \mathbb{R}^{64}$.
The second is PGGAN \cite{karras2017progressive}, trained on CelebA-HQ dataset \cite{karras2017progressive} that generates a $1024 \times 1024$ image from a Gaussian random vector $\z \in \mathbb{R}^{512}$.
We use the official pre-trained models,
and for details on the models and their training procedures  we refer the reader to the original publications.
Note that previous works, which use generative models for solving inverse problems, have considered much simpler datasets, such as MNIST \cite{lecun1998gradient} or a small version of CelebA (downscaled to size $64 \times 64$), which perhaps do not demonstrate how severe the effect of mode collapse is.

The test-time procedure is done as follows, and is almost the same for the two models.
For CSGM we follow \cite{bora2017compressed} and optimize \eqref{Eq_cost_func1} using ADAM optimizer \cite{kingma2014adam} with learning rate (LR) of 0.1. We use 1600 iterations for BEGAN and 1800 iterations for PGGAN. The final $\z$, i.e. $\hat{\z}$, is chosen to be the one with minimal objective value $f(\z)$ along the iterations, and the CSGM recovery is $\hat{\x}=\G(\hat{\z})$.
Performing a post-processing BP step \eqref{Eq_csgm_proj} %\eqref{Eq_csgm_proj}
gives us also a reconstruction that we denote by CSGM-BP.

In the reconstruction based on image-adaptive GANs, which we denote by IAGAN, we initialize $\z$ with $\hat{\z}$, and then optimize \eqref{Eq_cost_func_IA} jointly for $\z$ and $\btheta$ (the generator parameters).
%\textcolor{green}{
For BEGAN we use LR of $10^{-4}$ for both $\z$ and $\btheta$ in all scenarios, and for PGGAN we use LR of $10^{-4}$ and $10^{-3}$ for $\z$ and $\btheta$, respectively. For BEGAN, we use 600 iterations for compressed sensing and 500 for super-resolution. For PGGAN we use 500 and 300 iterations for compressed sensing and super-resolution, respectively.
In the examined noisy scenarios we use only half of the amount of iterations, to avoid overfitting the noise.
The final minimizers $\hat{\btheta}_{IA}$ and $\hat{\z}_{IA}$ are chosen according to the minimal objective value, %These settings are common for the two models.
and the IAGAN result is obtained by $\hat{\x}_{IA}=\G_{\hat{\btheta}_{IA}}(\hat{\z}_{IA})$.
Another recovery, which uses also the post-processing BP step \eqref{Eq_csgm_proj} %\eqref{Eq_csgm_proj}
on $\hat{\x}_{IA}$, is denoted by IAGAN-BP.

We also compare the methods with DIP \cite{ulyanov2017deep}. We use DIP official implementation for the noiseless scenarios, and for the examined noisy scenarios we reduce the number of iterations by a factor of 4 (tuned for best average performance) to prevent the network from overfitting the noise.

% Apart from presenting visual results\footnote{More visual examples are presented in %the %supp. material
% %{companion technical report \cite{shady2019image}.}}},
% Figures \ref{app:cs_03_noisy} - \ref{app:cs_05_noisy} for compressed sensing and in Figures \ref{app:sr_8} - \ref{app:sr_16} for super-resolution.},
% w
We compare the performance of the different methods using two quantitative measures. The first one is the widely-used mean squared error (MSE) %(sometimes in its PSNR form).
(sometimes in its PSNR form\footnote{\textbf{We compute the average PSNR as $10\mathrm{log}_{10}(255^2/\overline{\mathrm{MSE}})$, where $\overline{\mathrm{MSE}}$ is averaged over the test images.}}).
The second is a distance between images that focuses on perceptual similarity (PS), which has been proposed in \cite{zhang2018unreasonable} (we use the official implementation).
Displaying the PS is important since it is well known that PSNR/MSE may not correlate with the visual/perceptual quality of the reconstruction.
Note that in the PS score --- lower is better.

\begin{figure}[t]
  \centering
  \begin{subfigure}[b]{0.32\linewidth}
    %\centering\includegraphics[width=75pt]{figs/imgHQ00064_target_cs_0p3}
    \centering\includegraphics[width=75pt]{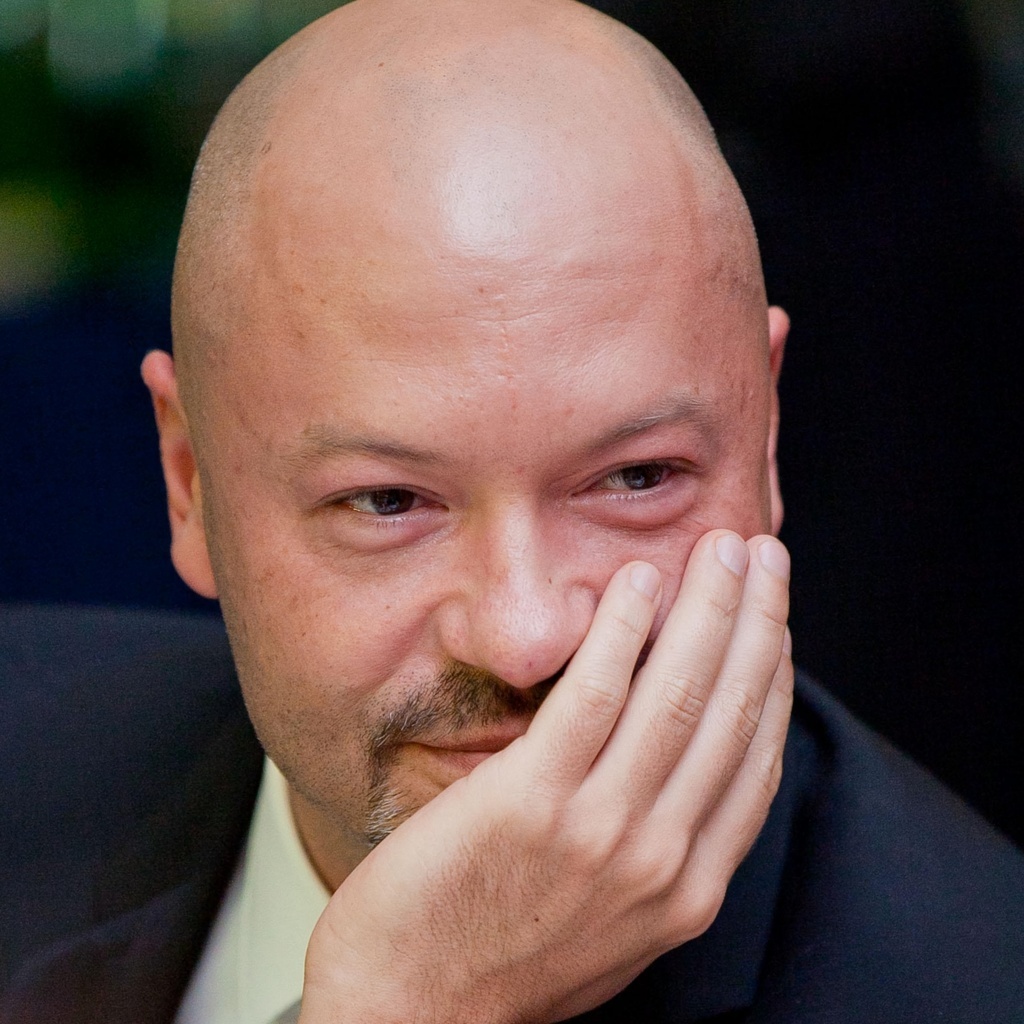}
  \end{subfigure}
  \begin{subfigure}[b]{0.32\linewidth}
    %\centering\includegraphics[width=75pt]{figs/imgHQ00064_compressed_cs_0p3}
    \centering\includegraphics[width=75pt]{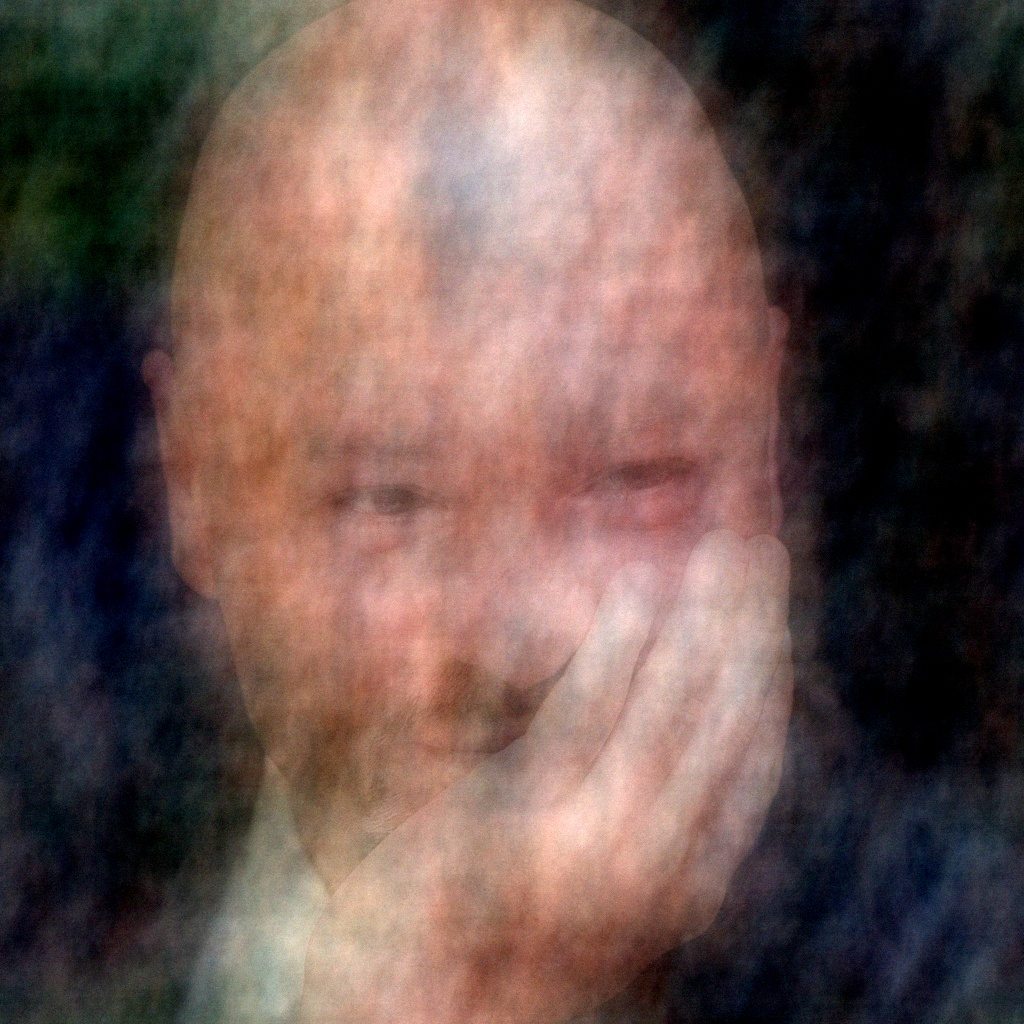}
  \end{subfigure}
  \\
    \begin{subfigure}[b]{0.32\linewidth}
    \centering\includegraphics[width=75pt]{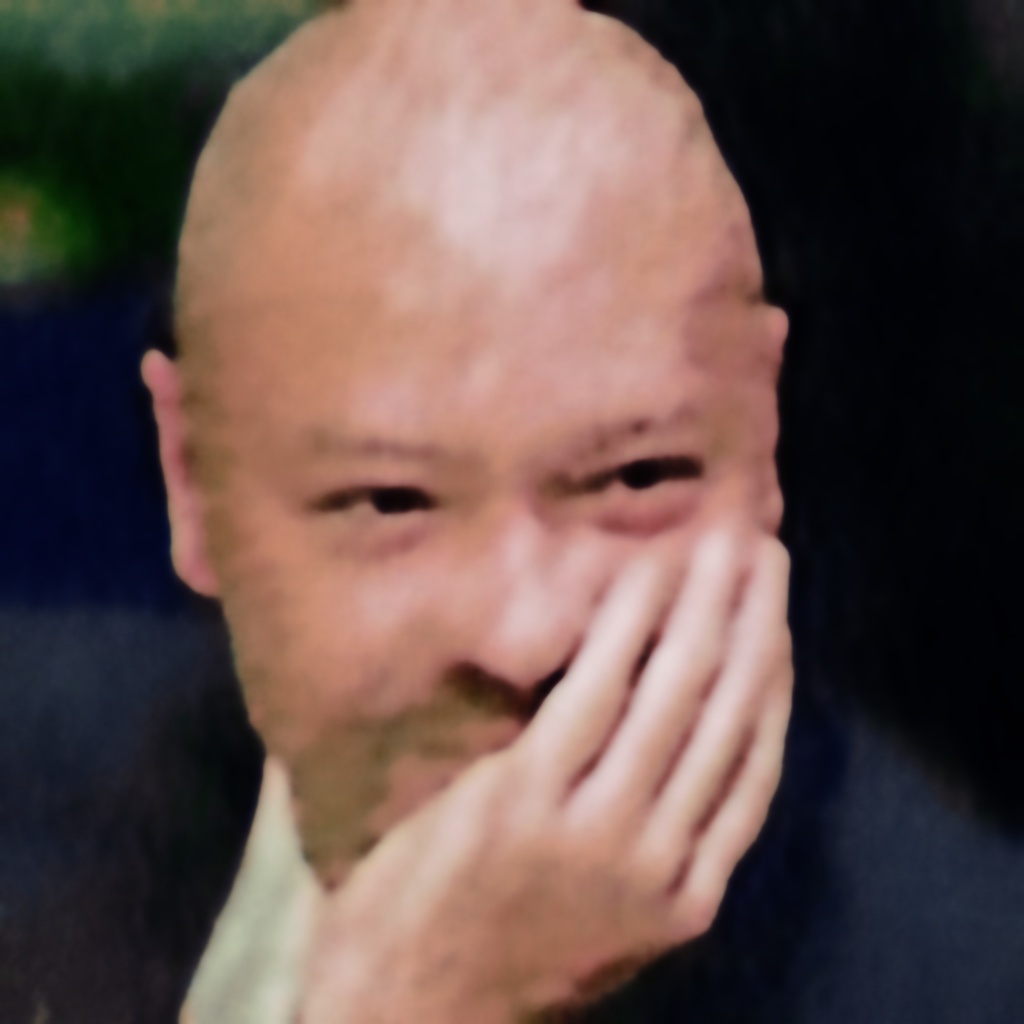}
  \end{subfigure}
  \begin{subfigure}[b]{0.32\linewidth}
    %\centering\includegraphics[width=75pt]{figs/imgHQ00064_CSGM_cs_0p3}
    \centering\includegraphics[width=75pt]{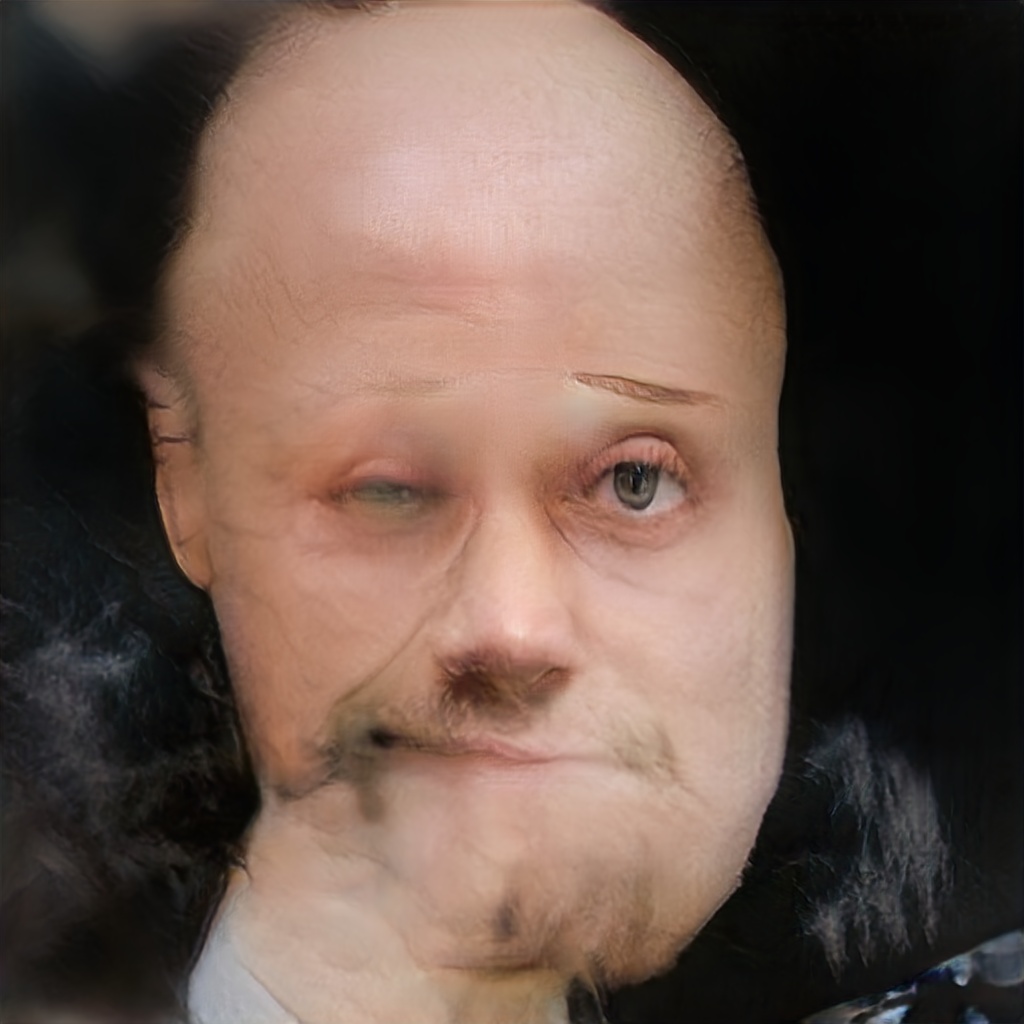}
  \end{subfigure}
  \begin{subfigure}[b]{0.32\linewidth}
    %\centering\includegraphics[width=75pt]{figs/imgHQ00064_IA_cs_0p3}
    \centering\includegraphics[width=75pt]{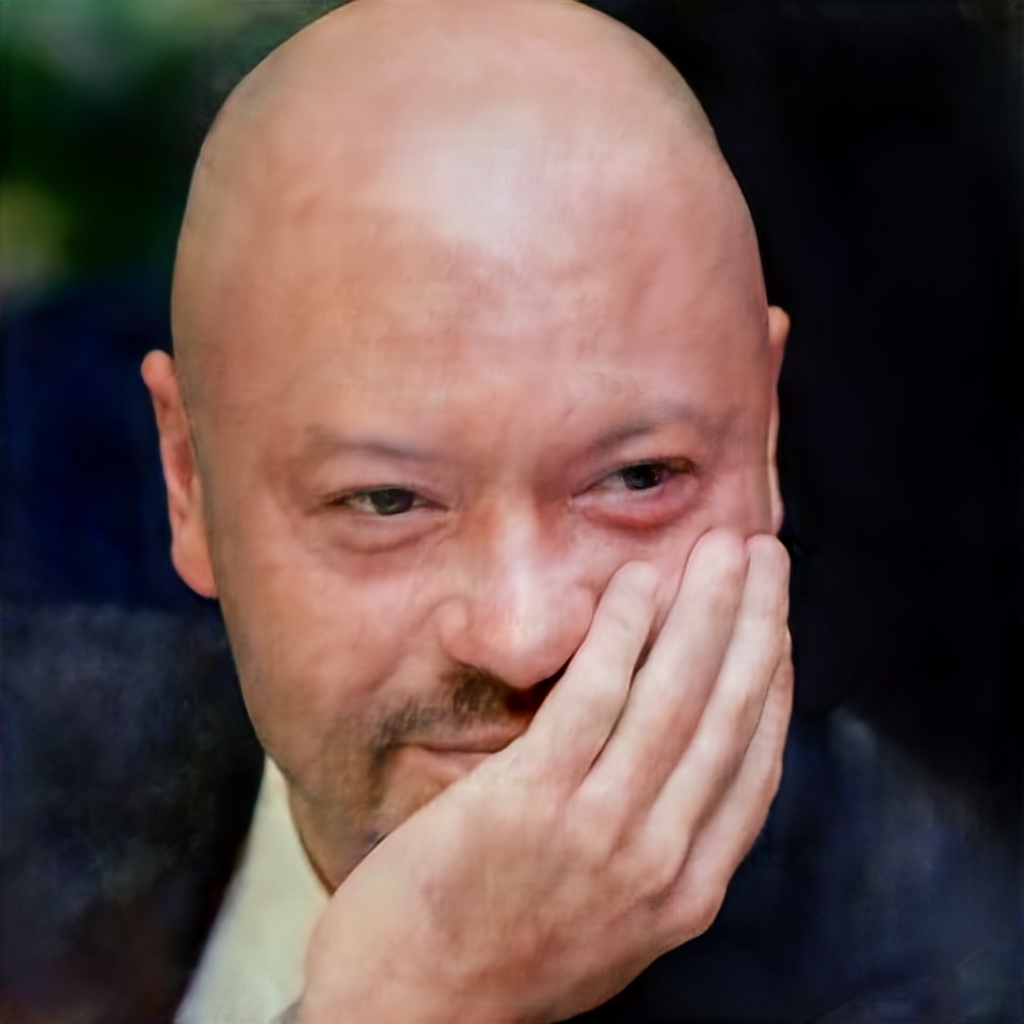}
  \end{subfigure}

% \vspace{1mm}

%   \begin{subfigure}[b]{0.32\linewidth}
%   %\vspace{1mm}
%     %\centering\includegraphics[width=75pt]{figs/imgHQ00075_target_cs_0p5}
%     \centering\includegraphics[width=75pt]{figs_aaai/imgHQ00075_target.jpg}
%   \end{subfigure}
%   \begin{subfigure}[b]{0.32\linewidth}
%     %\centering\includegraphics[width=75pt]{figs/imgHQ00075_compressed_cs_0p5}
%     \centering\includegraphics[width=75pt]{figs_aaai/imgHQ00075_compressed.jpg}
%   \end{subfigure}
%   \\
%   \begin{subfigure}[b]{0.32\linewidth}
%     \centering\includegraphics[width=75pt]{figs_aaai/imgHQ00075_DIP_loss0p00234911_scale=8p00_proper.jpg}
%   \end{subfigure}
%   \begin{subfigure}[b]{0.32\linewidth}
%     %\centering\includegraphics[width=75pt]{figs/imgHQ00075_CSGM_cs_0p5}
%     \centering\includegraphics[width=75pt]{figs_aaai/imgHQ00075_CSGM_loss=0p00594620.jpg}
%   \end{subfigure}
%   \begin{subfigure}[b]{0.32\linewidth}
%     %\centering\includegraphics[width=75pt]{figs/imgHQ00075_IA_cs_0p5}
%     \centering\includegraphics[width=75pt]{figs_aaai/imgHQ00075_IA_loss0p00102846.jpg}
%   \end{subfigure}
  \caption{Compressed sensing with 30\% subsampled Fourier measurements and noise level of 10/255, for CelebA-HQ images. From left to right and top to bottom: original image, naive reconstruction (zero padding and IFFT), DIP, CSGM, and IAGAN. Note that CSGM and IAGAN use the PGGAN prior.}
\label{fig:cs_0p5_Fourier_visual_hq}
\end{figure}

\begin{figure}
  \centering
  \begin{subfigure}[b]{0.4\linewidth}
    \centering\includegraphics[width=52pt]{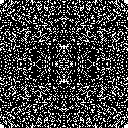}
  \end{subfigure}
  \begin{subfigure}[b]{0.4\linewidth}
    \centering\includegraphics[width=52pt]{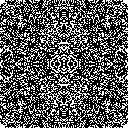}
  \end{subfigure}
  \caption{Binary masks for compressed sensing with 30\% (left) and 50\% (right) subsampled Fourier measurements.}
\label{fig:cs_masks}
\end{figure}

\begin{figure}[ht]
  \centering
  \begin{subfigure}[b]{0.32\linewidth}
    \centering\includegraphics[width=\linewidth]{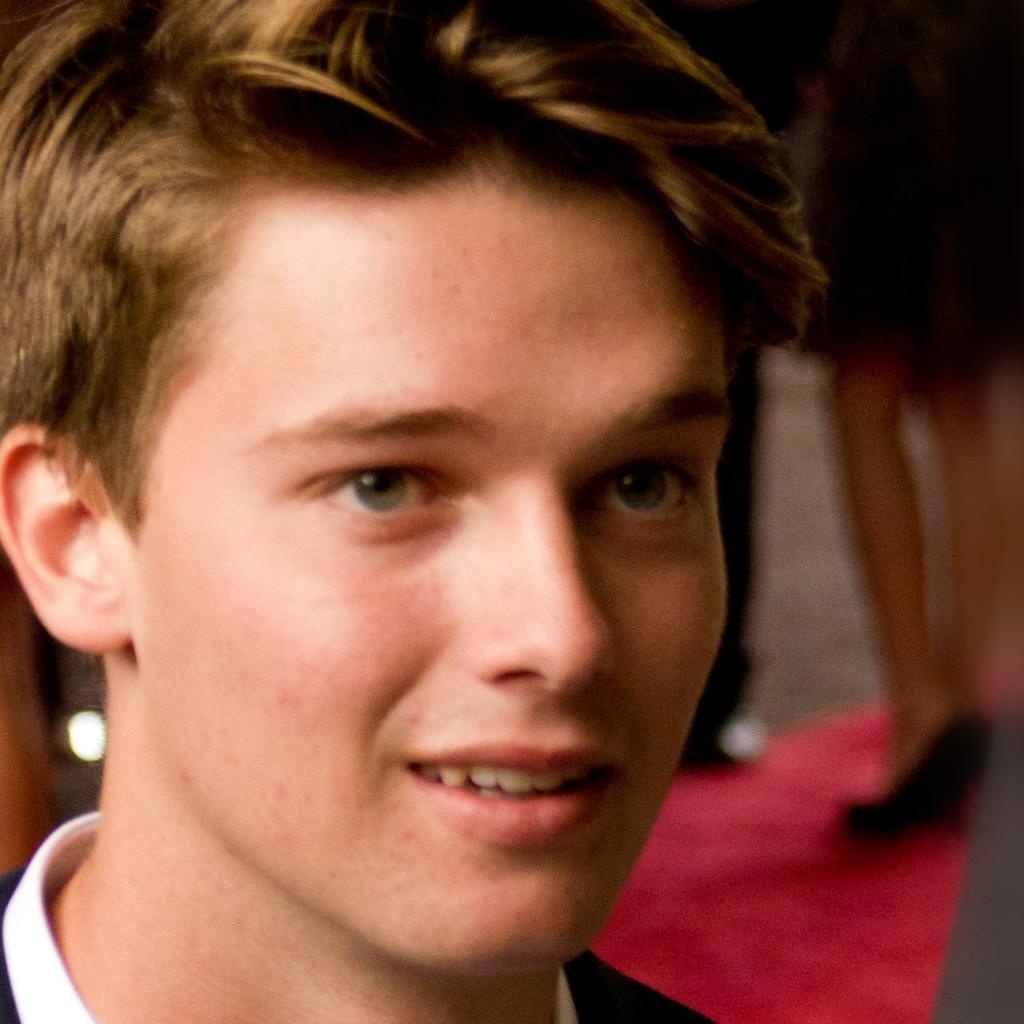}
  \end{subfigure}
  \begin{subfigure}[b]{0.32\linewidth}
    \centering\includegraphics[width=\linewidth]{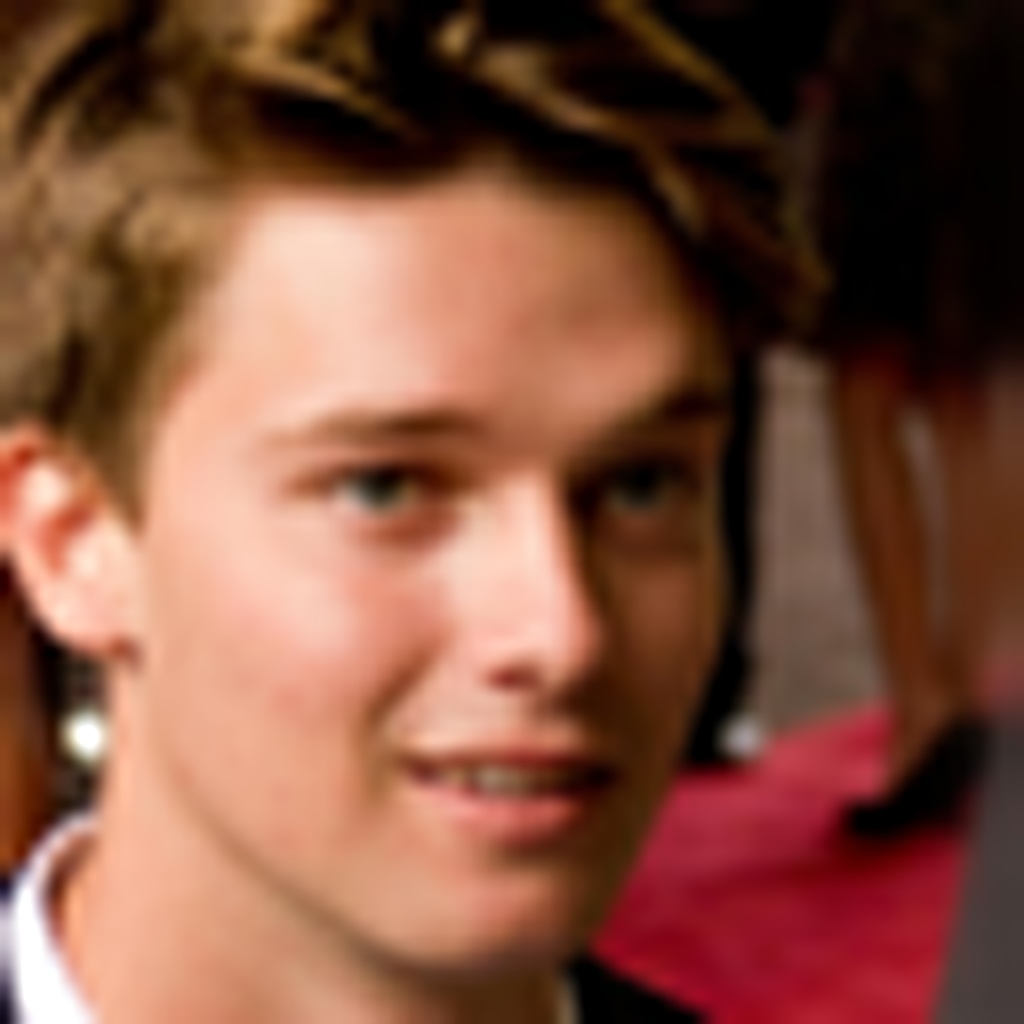}
  \end{subfigure}
  \\
  \begin{subfigure}[b]{0.32\linewidth}
    \centering\includegraphics[width=\linewidth]{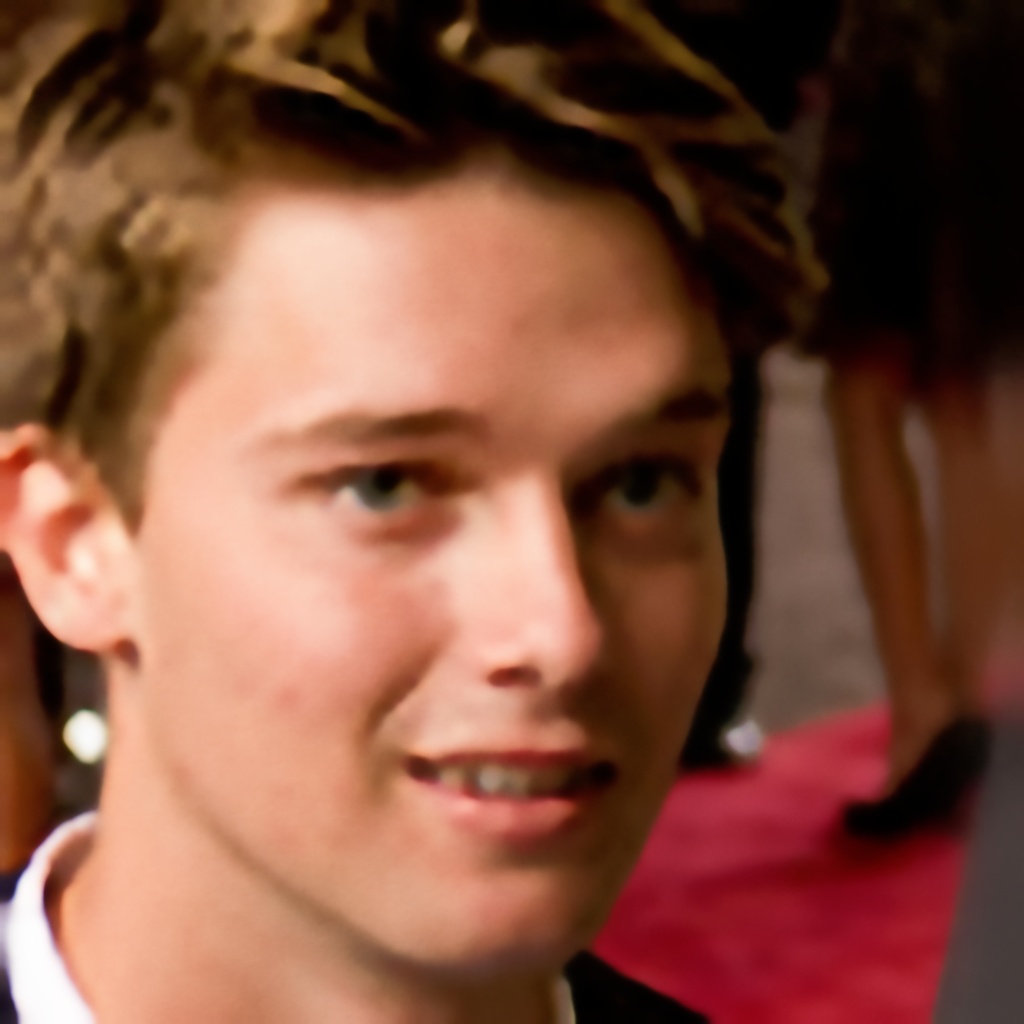}
  \end{subfigure}
  \begin{subfigure}[b]{0.32\linewidth}
    \centering\includegraphics[width=\linewidth]{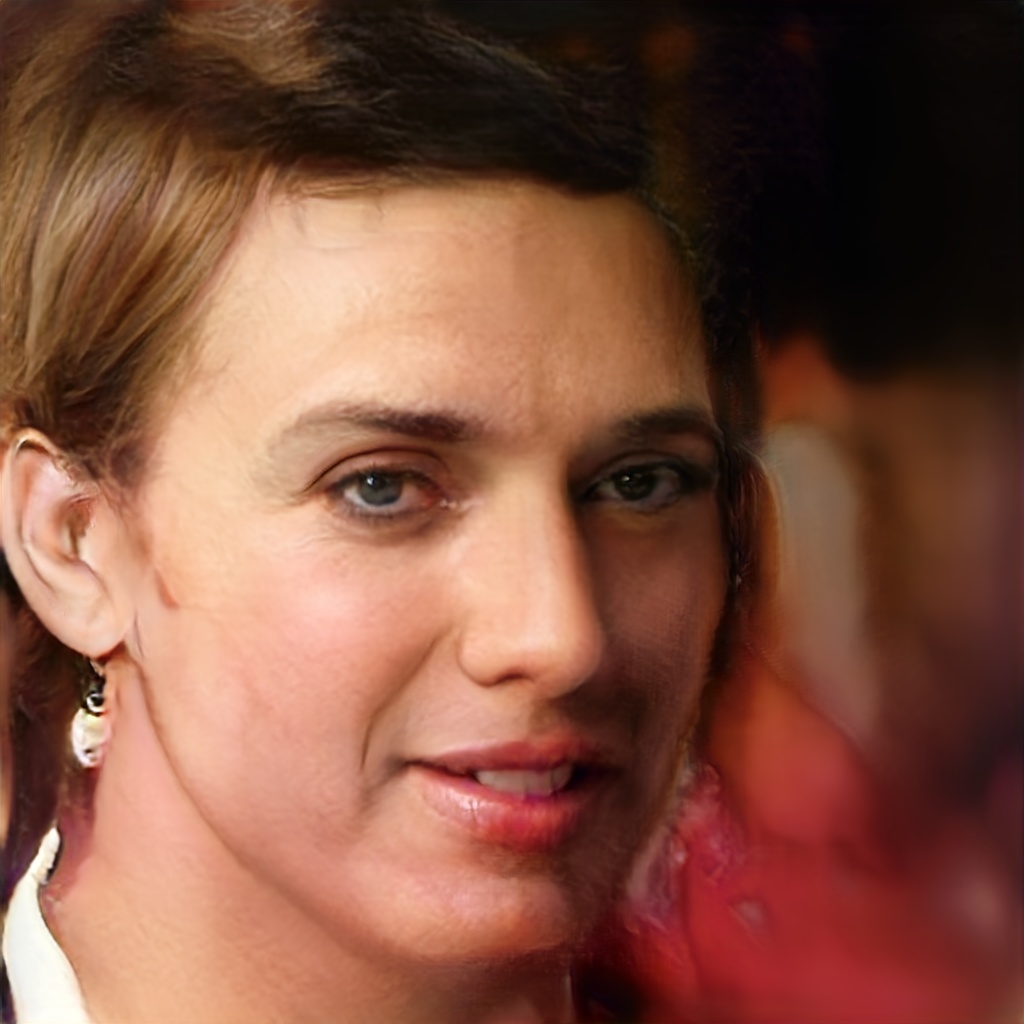}
  \end{subfigure}
  \begin{subfigure}[b]{0.32\linewidth}
    \centering\includegraphics[width=\linewidth]{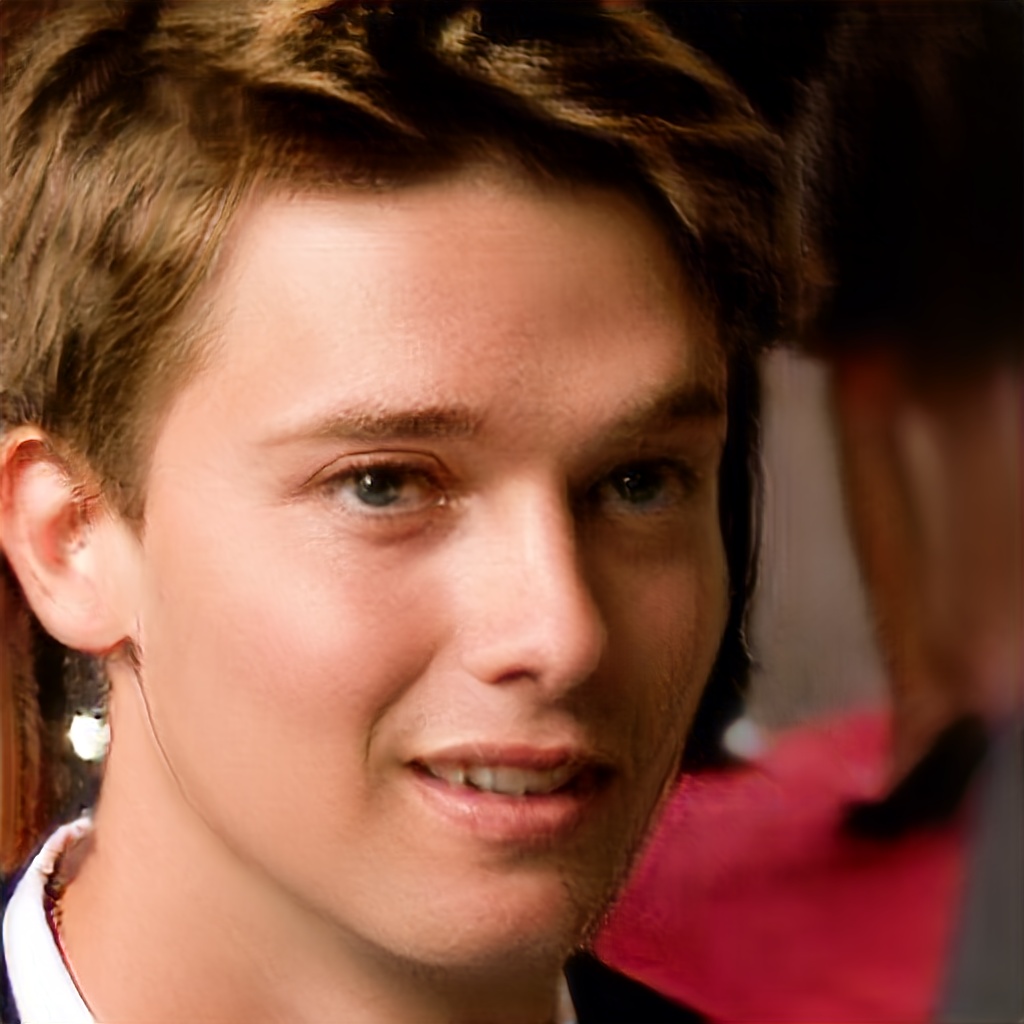}
  \end{subfigure}
  \\
  \begin{subfigure}[b]{0.32\linewidth}
    \centering\includegraphics[width=\linewidth]{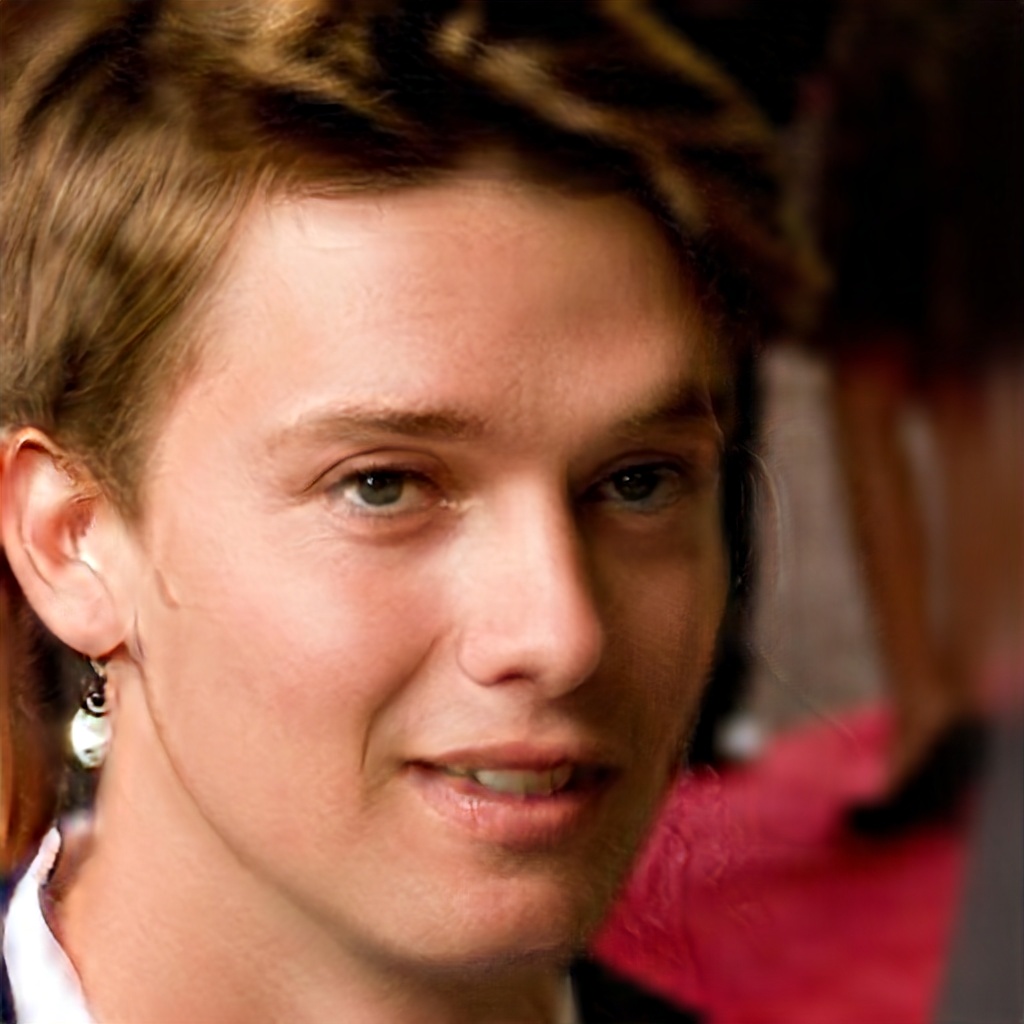}
  \end{subfigure}
  \begin{subfigure}[b]{0.32\linewidth}
    \centering\includegraphics[width=\linewidth]{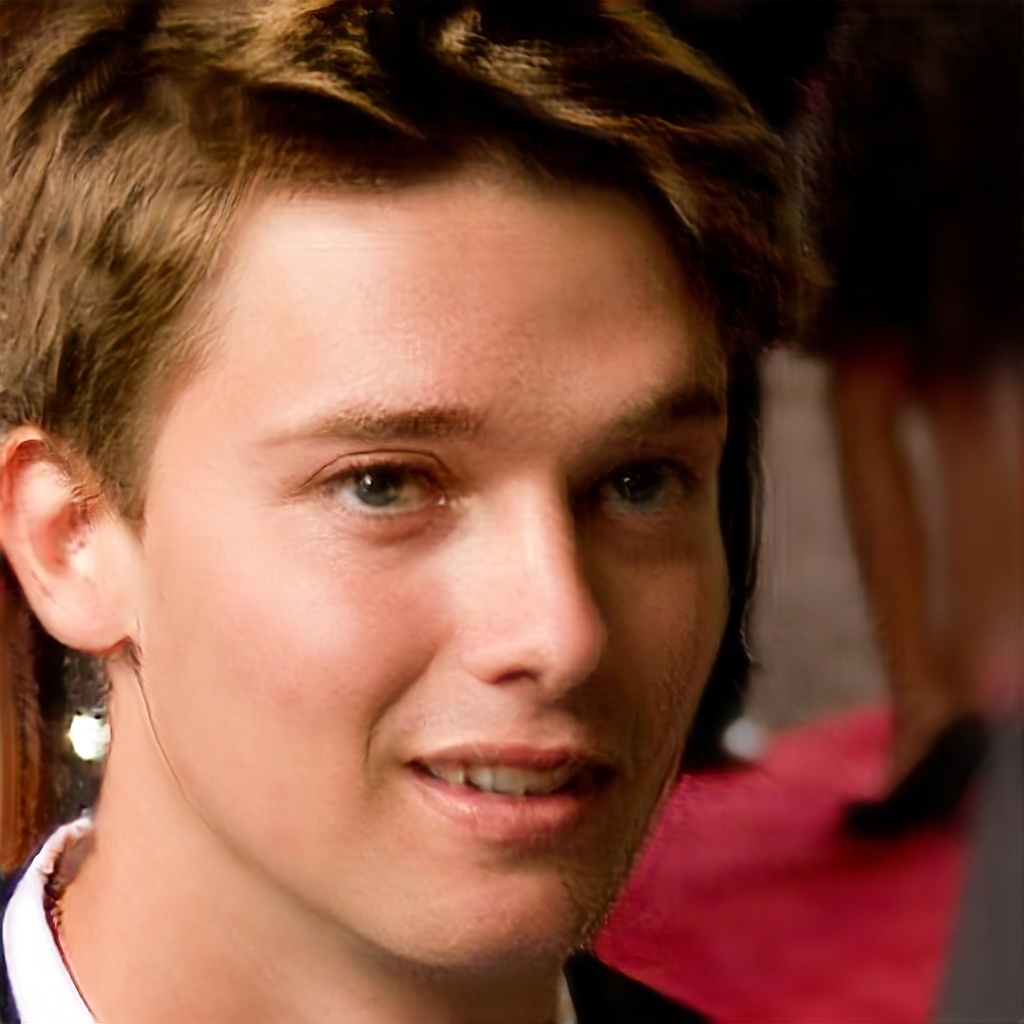}
  \end{subfigure}
  \caption{Super-resolution using PGGAN with scale factor of 16 without noise. From left to right and top to bottom: original image, bicubic upsampling, DIP, CSGM, IAGAN, CSGM-BP, and IAGAN-BP.}
\label{app:sr_16}
\end{figure}

\subsection{Compressed Sensing}

In the first experiment we demonstrate how the proposed IA and BP techniques significantly outperform or improve upon CSGM for a large range of compression ratios.
We consider noiseless compressed sensing using an $m \times n$ Gaussian %measurement
matrix $\A$ with i.i.d. entries drawn from $A_{ij} \sim \mathcal{N}(0,1/m)$, similar to the experiments in \cite{bora2017compressed}. In this case, there is no efficient way to implement the operators $\A$ and $\A^T$. Therefore, we consider only the BEGAN that generates $128 \times 128$ images (i.e. $n=3 \times 128^2=49,152$), which are much smaller than those generated by PGGAN.

Figure \ref{fig:cs_mse_vs_m_visual} shows several visual results and
Figure \ref{fig:cs_mse_vs_m} presents the reconstruction
MSE of the different methods
as we change the number of measurements $m$ (i.e. we change the compression ratio $m/n$). The results are averages over 20 images from CelebA dataset.
It is clearly seen that IAGAN outperforms CSGM for all the values of $m$. Note that due to the limited representation capabilities of BEGAN (equivalently -- its mode collapse), CSGM performance reaches a plateau in a quite small value of $m$, contrary to IAGAN error that continues to decrease.
The back-projection strategy is shown to be very effective, as it makes sure that CSGM-BP is rescued from the plateau of CSGM.
The fact that IAGAN still has a very small error when the compression ratio is almost 1 follows from our small learning rates and early stopping, which have been found necessary for small values of $m$, where the null space of $\A$ is very large and it is important to avoid overriding the offline semantic information. However, this small error is often barely visible, as demonstrated by visual results in Figure \ref{fig:cs_mse_vs_m_visual}, and further decreases by the BP step of IAGAN-BP.

\begin{table}%[t]
\scriptsize
\renewcommand{\arraystretch}{1.3}
\caption{Compressed sensing with subsampled Fourier measurements. Reconstruction PSNR [dB] (left) and PS \cite{zhang2018unreasonable} (right), averaged over 100 images from CelebA and CelebA-HQ, for compression ratios 0.3 and 0.5, with noise level of 10/255.} \label{table:cs_fourier}
\centering
    \begin{tabular}{ | l | l | l | l | l | l |}
    \hline
    {\em CelebA}        & naive IFFT & DIP & CSGM & IAGAN \\ \hline
    CS ratio 0.3 & 19.23 / 0.540 & {\bfseries 25.96} / 0.139 &  20.12 / 0.246 &  25.50 / {\bfseries 0.092}  \\ \hline
    CS ratio 0.5 & 20.53 / 0.495 &  27.21 / 0.125 & 20.32 / 0.241 &  {\bfseries 27.59} / {\bfseries 0.066}  \\ \hline
    \end{tabular}
    \begin{tabular}{ | l | l | l | l | l | l |}
    \hline
    {\em CelebA-HQ}        & naive IFFT & DIP & CSGM & IAGAN \\ \hline
    CS ratio 0.3 & 19.65 / 0.625 & 24.97 / 0.566 &  21.38 / 0.520 & {\bfseries 25.80} / {\bfseries 0.429}  \\ \hline
    CS ratio 0.5 & 20.45 / 0.597 & 26.29 / 0.535 & 21.82 / 0.514 & {\bfseries 28.26} / {\bfseries 0.378} \\ \hline
    \end{tabular}
\end{table}

\begin{table*}%[t]
\footnotesize
\renewcommand{\arraystretch}{1.3}
\caption{Super-resolution with bicubic downscaling kernel. Reconstruction PSNR [dB] (left) and PS \cite{zhang2018unreasonable} (right), averaged over 100 images from CelebA and CelebA-HQ, for scale factors 4, 8 and 16, with no noise.} \label{table:super_res}
\centering
    \begin{tabular}{ | l | l | l | l | l | l | l | l |}
    \hline
    {\em CelebA} \hspace{4.8mm}       & Bicubic & DIP & CSGM & CSGM-BP & IAGAN & IAGAN-BP \\ \hline
    SR x4        & 26.50 / 0.165 & {\bfseries 27.35} / 0.159 & 20.51 / 0.235 & 26.44 / 0.165 &  27.16 / {\bfseries 0.092} & 27.14 / {\bfseries 0.092}  \\ \hline
    SR x8        & 22.39 / 0.212 & 23.45 / 0.339 & 20.23 / 0.240 & 22.71 / 0.212 &  23.49 / 0.158 & {\bfseries 23.53} / {\bfseries 0.157}  \\ \hline
    \end{tabular}
    \\
    \begin{tabular}{ | l | l | l | l | l | l | l | l |}
    \hline
    {\em CelebA-HQ}        & Bicubic & DIP & CSGM & CSGM-BP & IAGAN & IAGAN-BP \\ \hline
    SR x8        & 29.94 / 0.398 & {\bfseries 30.01} / 0.400 & 22.62 / 0.505 & 28.54 / 0.398 & 28.76 / 0.387 & 28.76 / {\bfseries 0.360}  \\ \hline
    SR x16       & 27.43 / 0.437 & {\bfseries 27.51} / 0.480 & 22.34 / 0.506 & 26.20 / 0.437 & 26.28 / 0.421 & 25.86 / {\bfseries 0.411}  \\ \hline
    \end{tabular}
\end{table*}

\begin{table}%[t]
\scriptsize
\renewcommand{\arraystretch}{1.3}
\caption{Super-resolution with bicubic downscaling kernel. Reconstruction PSNR [dB] (left) and PS \cite{zhang2018unreasonable} (right), averaged over 100 images from CelebA and CelebA-HQ, for scale factors 4, 8 and 16, with noise level of 10/255.} \label{table:super_res_noisy}
\centering
    \begin{tabular}{ | l | l | l | l | l | l |}
    \hline
    {\em CelebA} \hspace{3.8mm}       & Bicubic & DIP & CSGM & IAGAN  \\ \hline
    SR x4        & 24.72 / 0.432 & 24.19 / 0.280 & 20.57 / 0.238 &  {\bfseries 25.54} / {\bfseries 0.133}  \\ \hline
    SR x8        & 21.65 / 0.660 & 21.22 / 0.513 & 20.22 / {\bfseries 0.243} &  {\bfseries 21.72} / {\bfseries 0.243}  \\ \hline
    \end{tabular}
    \begin{tabular}{ | l | l | l | l | l | l |}
    \hline
    {\em CelebA-HQ}        & Bicubic & DIP & CSGM & IAGAN \\ \hline
    SR x8        & 26.31 / 0.801 & {\bfseries 27.61} / 0.430 & 21.60 / 0.519 & 26.30 / {\bfseries 0.421}  \\ \hline
    SR x16       & {\bfseries 25.02} / 0.781 & 24.20 / 0.669 & 21.31 / 0.516 & 24.73 / {\bfseries 0.455}  \\ \hline
    \end{tabular}
\end{table}

In order to examine our proposed IA strategy for the larger model PGGAN as well, we turn to use a different measurement operator $\A$ which can be applied efficiently -- the subsampled Fourier transform. This acquisition model is also more common in practice, e.g. in sparse MRI \cite{lustig2007sparse}.
We consider scenarios with compression ratios of 0.3 and 0.5, and noise level of 10/255 (due to the noise we do not apply the BP post-processing).
The PSNR and PS results (averaged on 100 images from each dataset) are given in Table \ref{table:cs_fourier}, and several visual examples are shown in Figure \ref{fig:cs_0p5_Fourier_visual_hq}.
In Figure \ref{fig:cs_masks} we
present the binary masks used for 30\% and 50\% Fourier domain sampling of $128 \times 128$ images in CelebA. The binary masks that have been used for CelebA-HQ have similar forms.

\begin{figure}
  \centering
  \begin{subfigure}[b]{0.23\linewidth}
    %\centering\includegraphics[width=55pt]{figs/000020_target_SR}
    \centering\includegraphics[width=55pt]{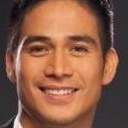}
  \end{subfigure}
  \begin{subfigure}[b]{0.23\linewidth}
    %\centering\includegraphics[width=55pt]{figs/000020_Bicubic_PIL_SR}
    \centering\includegraphics[width=55pt]{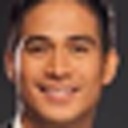}
  \end{subfigure}
  \begin{subfigure}[b]{0.23\linewidth}
    \centering\includegraphics[width=55pt]{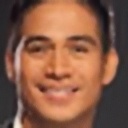}
  \end{subfigure}
  \\
  \begin{subfigure}[b]{0.23\linewidth}
    %\centering\includegraphics[width=55pt]{figs/000020_CSGM_rec_loss0p01505696_SR}
    \centering\includegraphics[width=55pt]{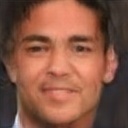}
  \end{subfigure}
  \begin{subfigure}[b]{0.23\linewidth}
    %\centering\includegraphics[width=55pt]{figs/000020_CSGM_NaiveBlend_Bicubic_PIL_SR}
    \centering\includegraphics[width=55pt]{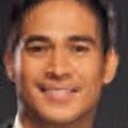}
  \end{subfigure}
  \begin{subfigure}[b]{0.23\linewidth}
    %\centering\includegraphics[width=55pt]{figs/000020_IA_loss0p00258777_SR}
    \centering\includegraphics[width=55pt]{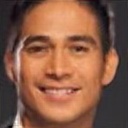}
  \end{subfigure}
  \begin{subfigure}[b]{0.23\linewidth}
    %\centering\includegraphics[width=55pt]{figs/000020_IA_NaiveBlend_Bicubic_PIL_SR}
    \centering\includegraphics[width=55pt]{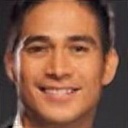}
  \end{subfigure}

  \caption{Super-resolution with scale factor 4, bicubic kernel, and no noise, for CelebA images. From left to right and top to bottom: original image, bicubic upsampling, DIP, CSGM, CSGM-BP, IAGAN, and IAGAN-BP. Note that CSGM and IAGAN use the BEGAN prior.}
\label{fig:SR_visual}
\end{figure}

The unsatisfactory results obtained by CSGM clearly demonstrate the limited capabilities of both BEGAN and PGGAN for
reconstruction:
%\textcolor{blue}{
Despite the fact that both of them can generate very nice samples \cite{berthelot2017began,karras2017progressive},
they typically cannot represent well an image that fits the given observations $\y$. This is resolved by our image-adaptive approach. %}
For CelebA dataset DIP has competitive PSNR with our IAGAN. However, both the qualitative examples and the PS (perceptual similarity) measure agree that IAGAN results are much more pleasing.
For CelebA-HQ dataset our IAGAN clearly outperforms the other methods.

%\textcolor{blue}{
\textbf{Inference run-time.} Since IAGAN performs a quite small number of ADAM iterations to jointly optimize $\z$ and $\btheta$ (the generator's parameters), it requires only a small additional time compared to CSGM. Yet, both methods are much faster than DIP, which trains from scratch a large CNN at test-time. For example, for compression ratio of 0.5, using NVIDIA RTX 2080ti GPU we got the following per image run-time:
for CelebA: DIP $\sim$100s, CSGM $\sim$30s, and IAGAN $\sim$35s; and
for CelebA-HQ: DIP $\sim$1400s, CSGM $\sim$120s, and IAGAN $\sim$140s.
The same behavior, i.e. CSGM and IAGAN are much faster than DIP, holds throughout the experiments in the paper (e.g. also for the super-resolution task).
%}

%\raggedbottom

\subsection{Super-Resolution}

We turn to examine the super-resolution task, for $\A$ which is a composite operator of blurring with a bicubic anti-aliasing kernel followed by down-sampling.
For BEGAN we use super-resolution scale factors of 4 and 8, and for PGGAN we use scale factors of 8 and 16. We check a noiseless scenario and a scenario with noise level of 10/255.
For the noiseless scenario we also examine the GAN-based recovery after a BP post-processing, which
%Similar to previous compressed sensing experiment, again the BP step
can be computed efficiently, because $\A^\dagger$ can be implemented by bicubic upsampling.
The PSNR and PS results (averaged on 100 images from each dataset) of the different methods are given in Tables \ref{table:super_res} and \ref{table:super_res_noisy}, and several visual examples are shown in Figure \ref{fig:SR_visual}.

Once again, the results of the plain CSGM are not satisfying.
Due to the limited representation capabilities of BEGAN and PGGAN, the recovered faces
look very different than the ones in the test images.
The BP post-processing is very effective in reducing CSGM representation error when the noise level is minimal.
For our IAGAN approach, the BP step is less effective (i.e. IAGAN and IAGAN-BP have similar recoveries), which implies that the "soft-compliance" of IAGAN obtains similar results as the "hard-compliance" of the BP in the row space of $\A$.
In the noiseless case, DIP often obtains better PSNR than IAGAN. However, as observed in the compressed sensing experiments, both the visual examples and the PS (perceptual similarity) measure agree that IAGAN results are much more pleasing and sharper, in both noisy and noiseless scenarios.
A similar tradeoff between distortion and perception has been recently investigated by \citet{blau2018perception}.
Their work supports the observation that the balance between fitting the measurements and preserving the generative prior, which is the core of our IAGAN approach, may limit the achievable PSNR in some cases but significantly improves the perceptual quality.

\begin{figure}[ht]
  \centering
  \begin{subfigure}[b]{0.23\linewidth}
    \centering\includegraphics[width=55pt]{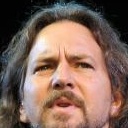}
  \end{subfigure}
  \begin{subfigure}[b]{0.23\linewidth}
    \centering\includegraphics[width=55pt]{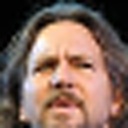}
  \end{subfigure}
  \begin{subfigure}[b]{0.23\linewidth}
    \centering\includegraphics[width=55pt]{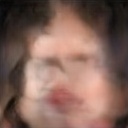}
  \end{subfigure}
  \begin{subfigure}[b]{0.23\linewidth}
    \centering\includegraphics[width=55pt]{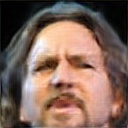}
  \end{subfigure}
  \\
  \begin{subfigure}[b]{0.23\linewidth}
    \centering\includegraphics[width=55pt]{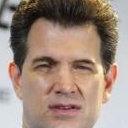}
  \end{subfigure}
  \begin{subfigure}[b]{0.23\linewidth}
    \centering\includegraphics[width=55pt]{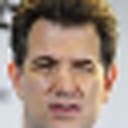}
  \end{subfigure}
  \begin{subfigure}[b]{0.23\linewidth}
    \centering\includegraphics[width=55pt]{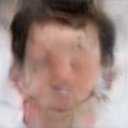}
  \end{subfigure}
  \begin{subfigure}[b]{0.23\linewidth}
    \centering\includegraphics[width=55pt]{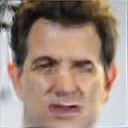}
  \end{subfigure}
  \caption{Super-resolution of misaligned images with bicubic kernel and scale factor 4 using BEGAN. Left to right: original image, bicubic upsampling, CSGM, and IAGAN.}
\label{fig:SR_misalign}
\end{figure}

We finish this section with an extreme demonstration of mode collapse. In this scenario we use the BEGAN model to
super-resolve images with scale factor of 4. Yet, this time the images are slightly misaligned --- they are vertically translated by a few pixels.
The PSNR[dB] / PS results (averaged on 100 CelebA images) are 19.18 / 0.374 for CSGM and 26.73 / 0.127 for IAGAN. Several visual results are shown in Figure \ref{fig:SR_misalign}.
The CSGM is highly susceptible to the poor capabilities of BEGAN in this case, while our IAGAN strategy is quite robust.

\subsection{Deblurring}

%\
We briefly demonstrate that the advantage of IAGAN carries to more inverse problems by examining a deblurring scenario, where the operator $\A$ represents blurring with a $9 \times 9$ uniform filter, and the noise level is 10/255 (so we do not apply the BP post-processing).
The PSNR and PS results (averaged on 100 images from each dataset) of DIP, CSGM, and IAGAN
are given in Table \ref{table:deblur_noisy}, and several visual examples are presented in
Figure \ref{fig:deb_visual_hq}.
%the companion technical report \cite{shady2019image}.
%}

%\textcolor{blue}{
Similarly to the previous experiments, the proposed IAGAN often exhibits the best PSNR and consistently exhibits the best perceptual quality.
%}

\begin{table}%[t]
\scriptsize
\renewcommand{\arraystretch}{1.3}
\caption{Deblurring with $9 \times 9$ uniform filter and noise level of 10/255. Reconstruction PSNR [dB] (left) and PS \cite{zhang2018unreasonable} (right), averaged over 100 images from CelebA and CelebA-HQ.} \label{table:deblur_noisy}
\centering
    \begin{tabular}{ | l | l | l | l | l | l |}
    \hline
    {\em CelebA} \hspace{3.8mm}       & Blurred & DIP & CSGM & IAGAN  \\ \hline
    Deb U(9x9)        &  22.21 / 0.490 & 25.63 / 0.203 & 20.37 / 0.241 &  {\bfseries 26.15} / {\bfseries 0.110}  \\ \hline
    \end{tabular}
    \begin{tabular}{ | l | l | l | l | l | l |}
    \hline
    {\em CelebA-HQ}        & Blurred & DIP & CSGM & IAGAN \\ \hline
    Deb U(9x9)        & 25.80 / 0.622 & {\bfseries 28.28} / 0.458 & 21.62 / 0.507 & 28.25 / {\bfseries 0.388}  \\ \hline
    \end{tabular}
\end{table}

\begin{figure}[t]
  \centering
  \begin{subfigure}[b]{0.32\linewidth}
    \centering\includegraphics[width=75pt]{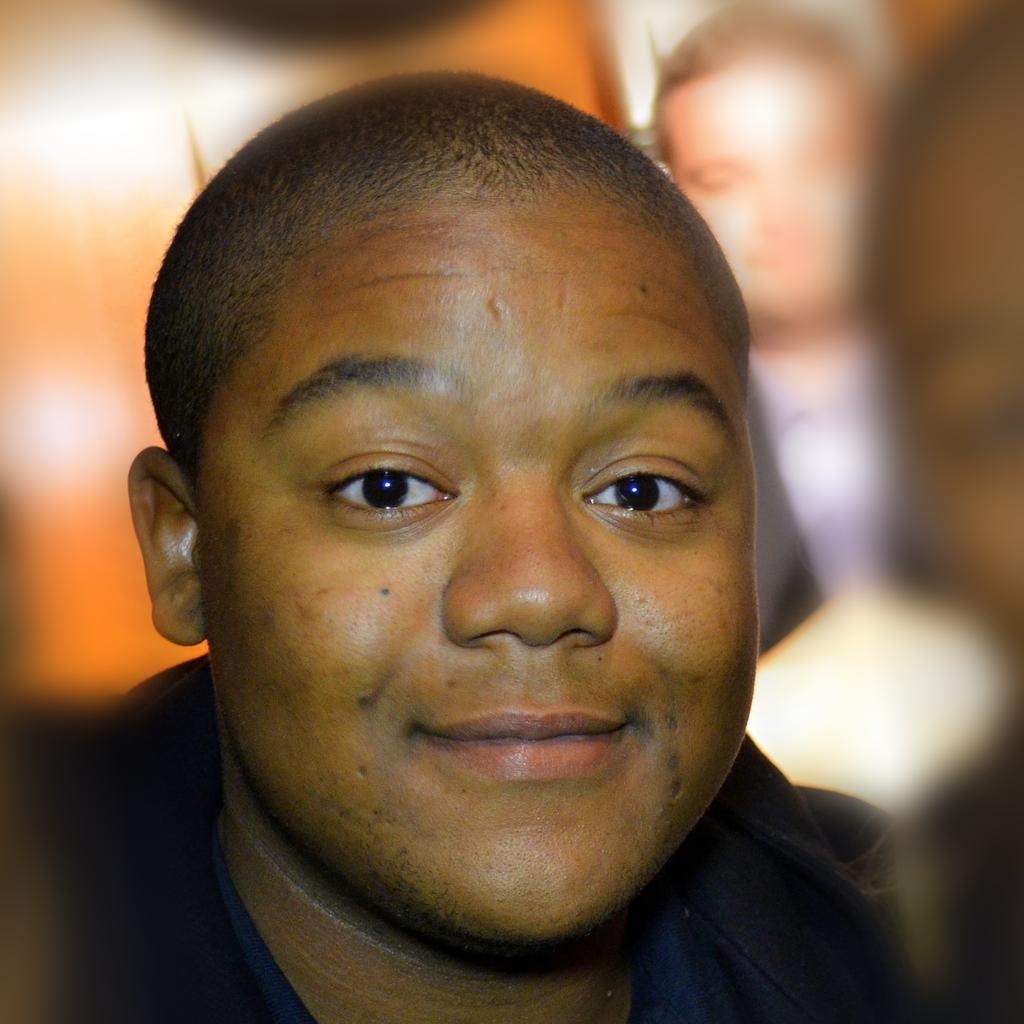}
  \end{subfigure}
  \begin{subfigure}[b]{0.32\linewidth}
    \centering\includegraphics[width=75pt]{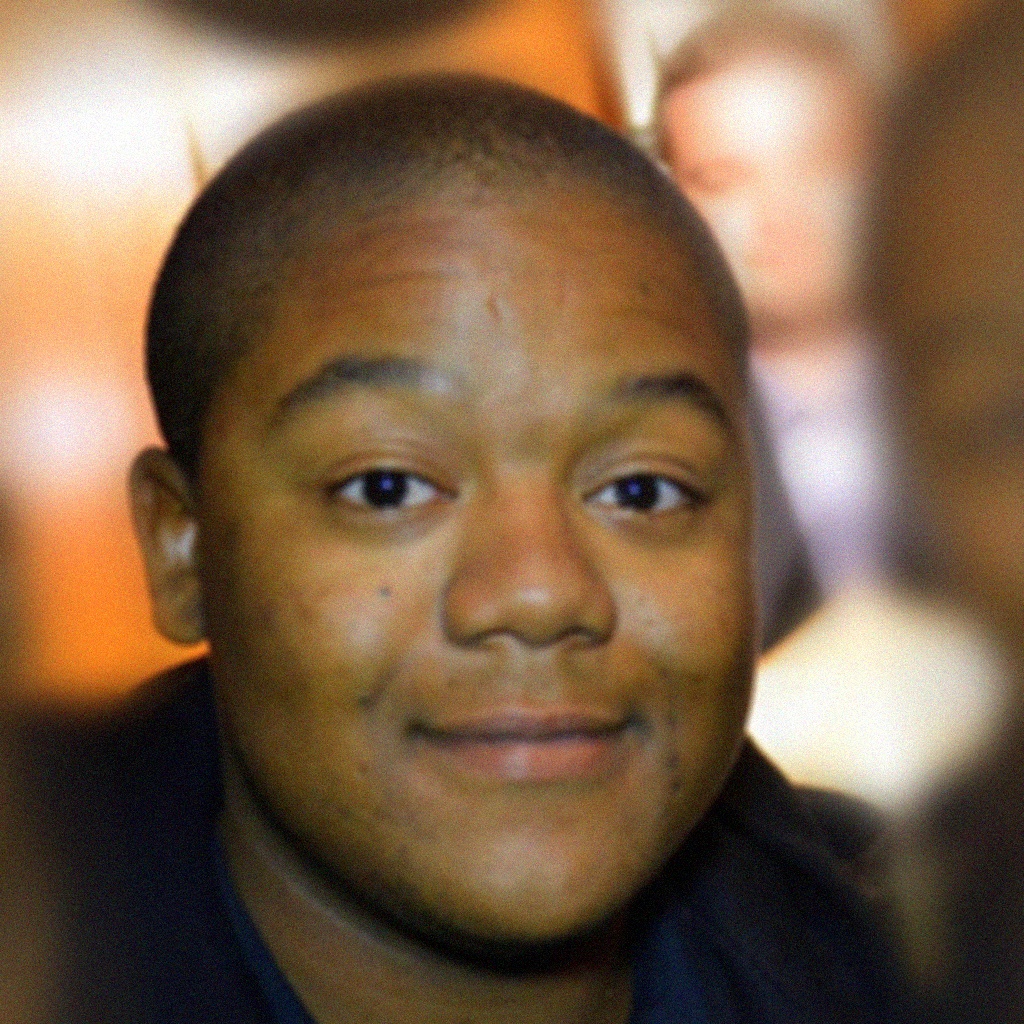}
  \end{subfigure}
  \\
    \begin{subfigure}[b]{0.32\linewidth}
    \centering\includegraphics[width=75pt]{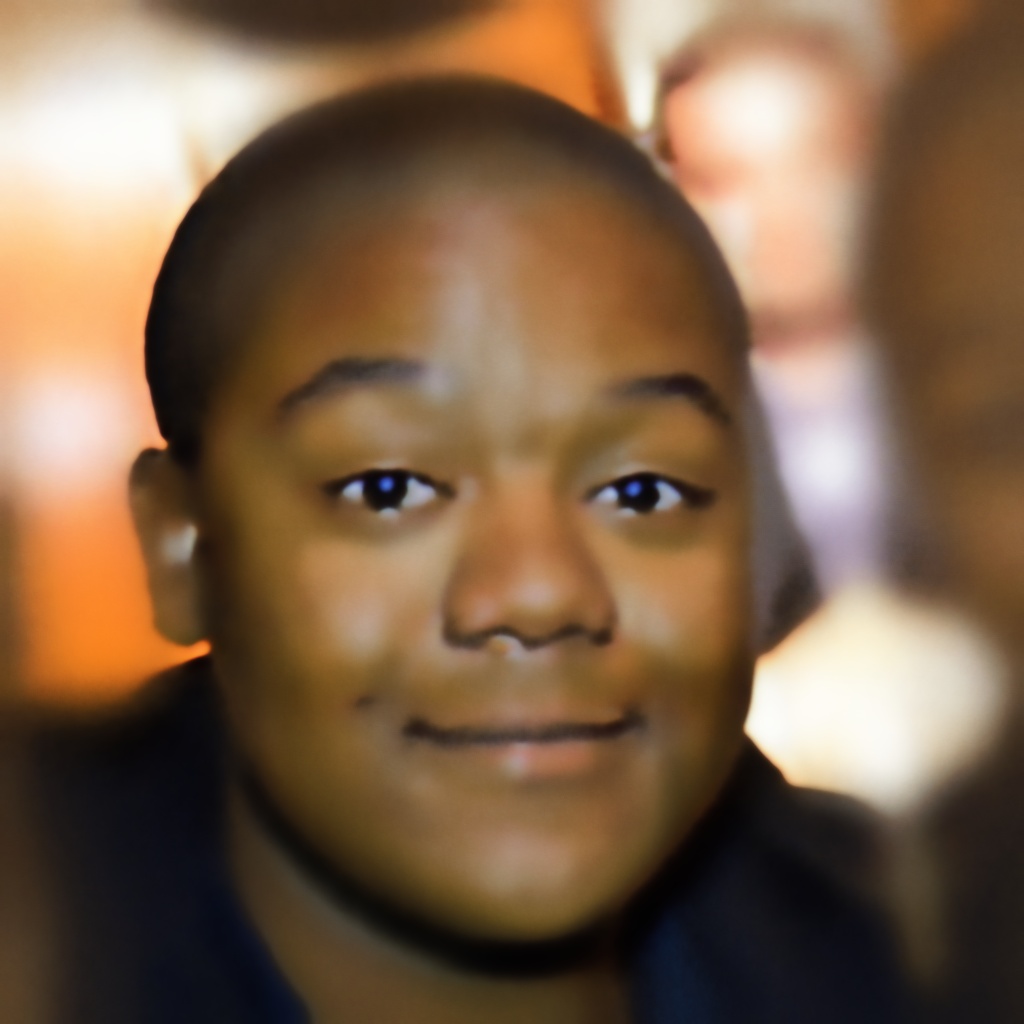}
  \end{subfigure}
  \begin{subfigure}[b]{0.32\linewidth}
    %\centering\includegraphics[width=75pt]{figs/imgHQ00064_CSGM_cs_0p3}
    \centering\includegraphics[width=75pt]{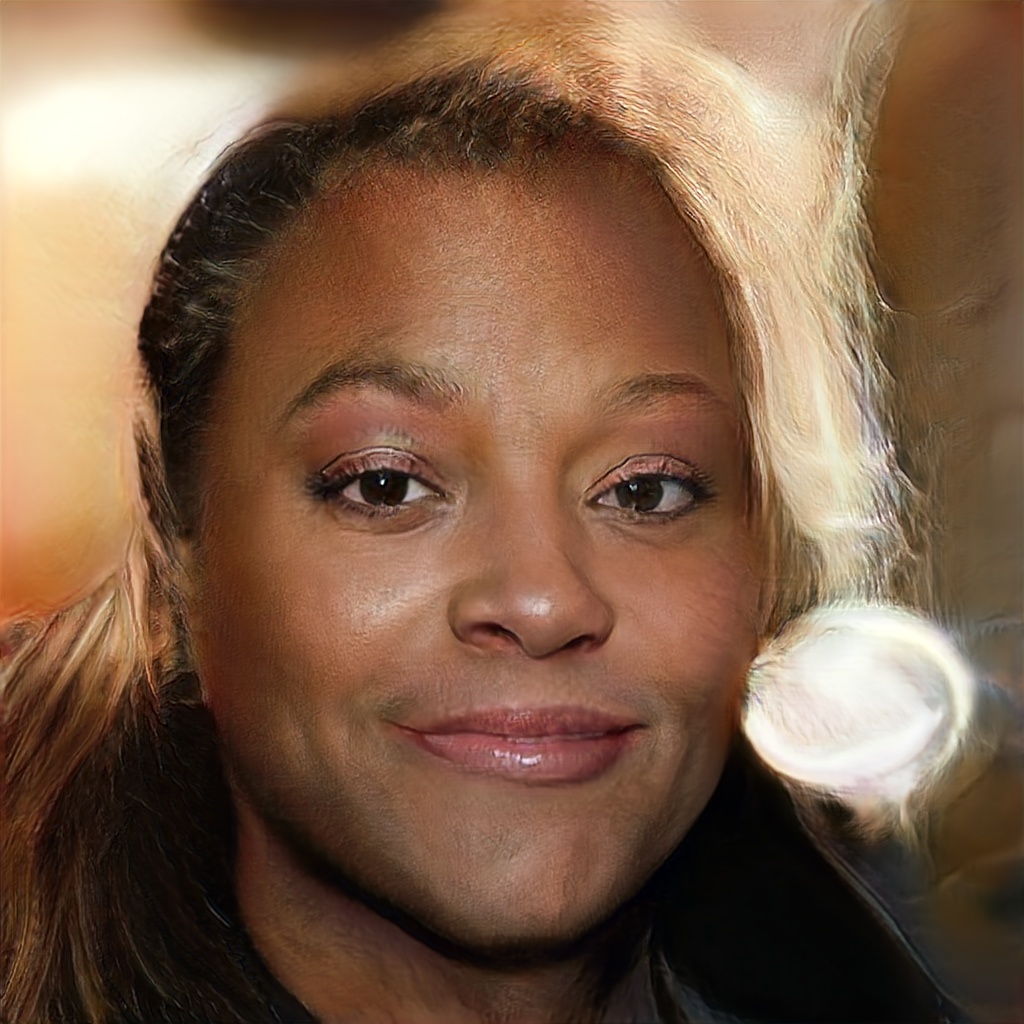}
  \end{subfigure}
  \begin{subfigure}[b]{0.32\linewidth}
    %\centering\includegraphics[width=75pt]{figs/imgHQ00064_IA_cs_0p3}
    \centering\includegraphics[width=75pt]{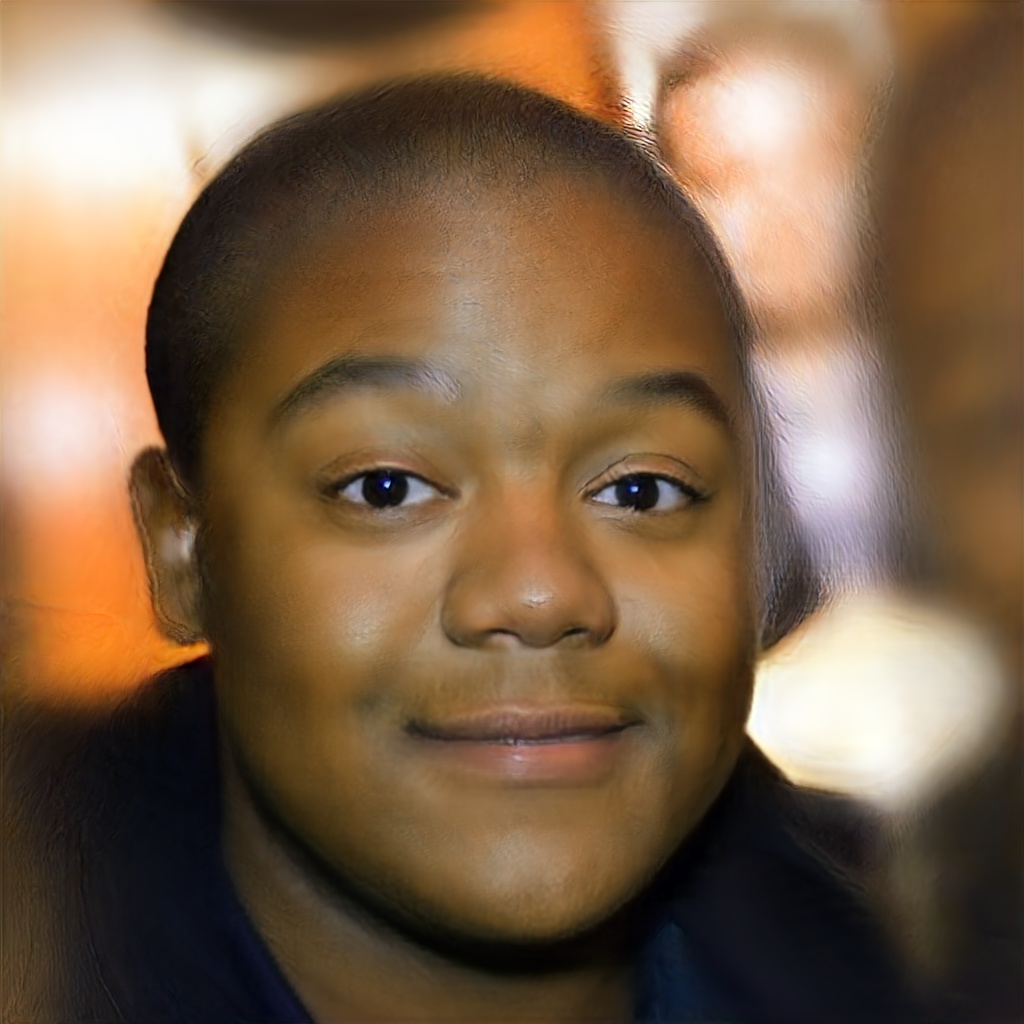}
  \end{subfigure}

\vspace{3mm}

  \begin{subfigure}[b]{0.32\linewidth}
  %\vspace{1mm}
    %\centering\includegraphics[width=75pt]{figs/imgHQ00075_target_cs_0p5}
    \centering\includegraphics[width=75pt]{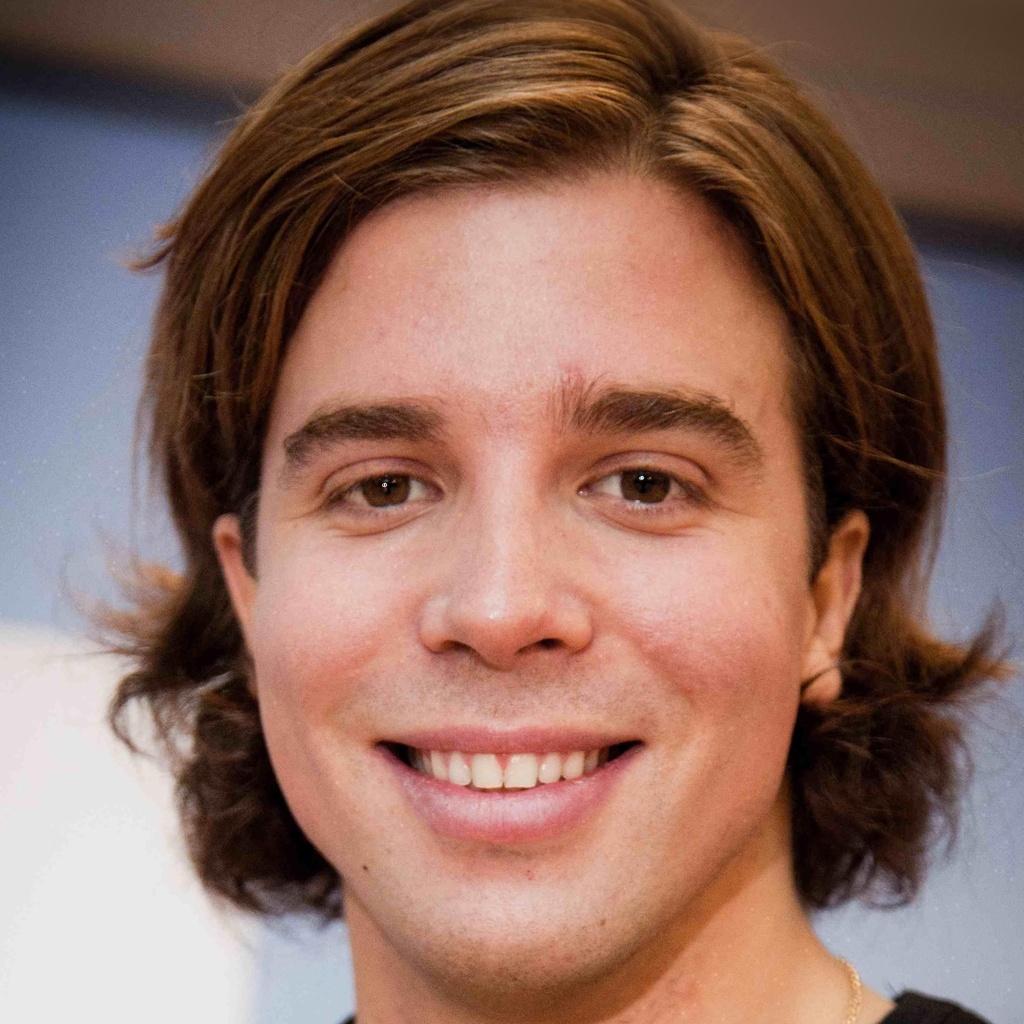}
  \end{subfigure}
  \begin{subfigure}[b]{0.32\linewidth}
    %\centering\includegraphics[width=75pt]{figs/imgHQ00075_compressed_cs_0p5}
    \centering\includegraphics[width=75pt]{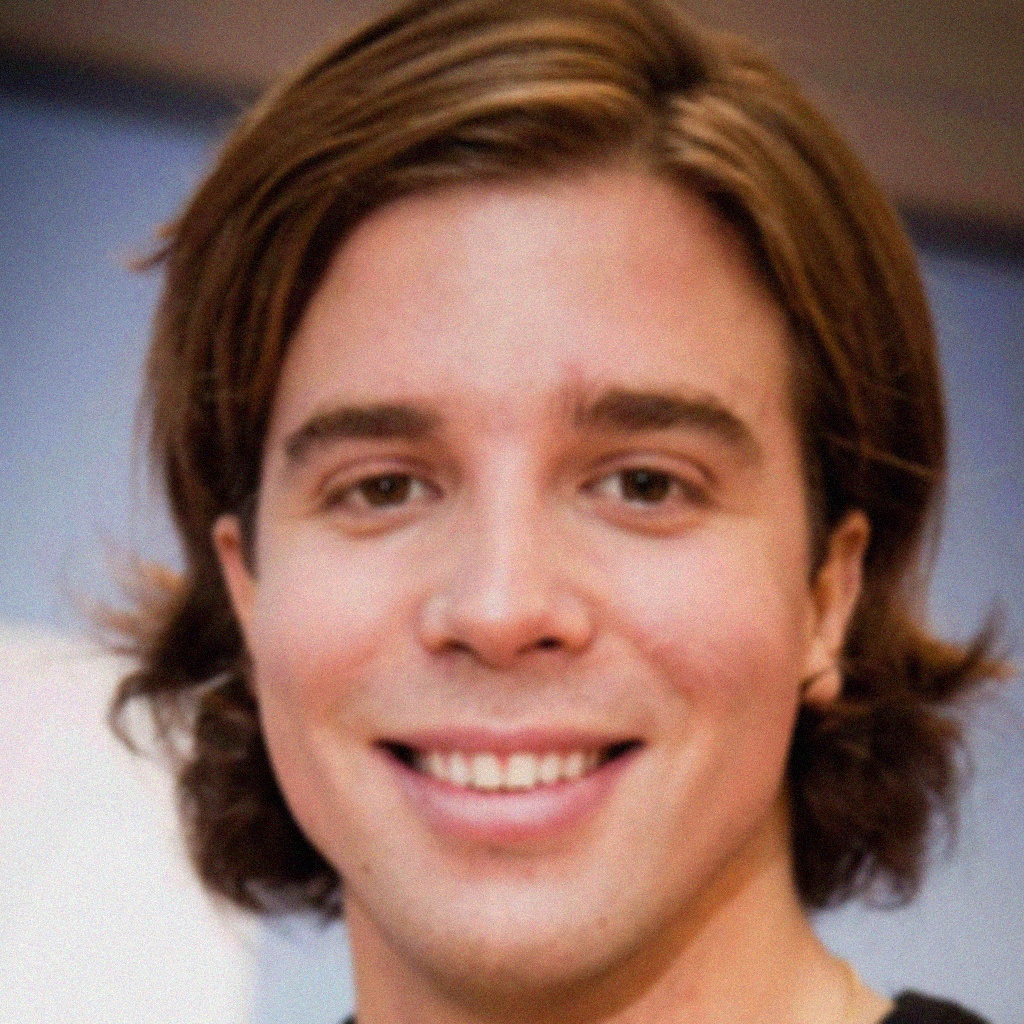}
  \end{subfigure}
  \\
  \begin{subfigure}[b]{0.32\linewidth}
    \centering\includegraphics[width=75pt]{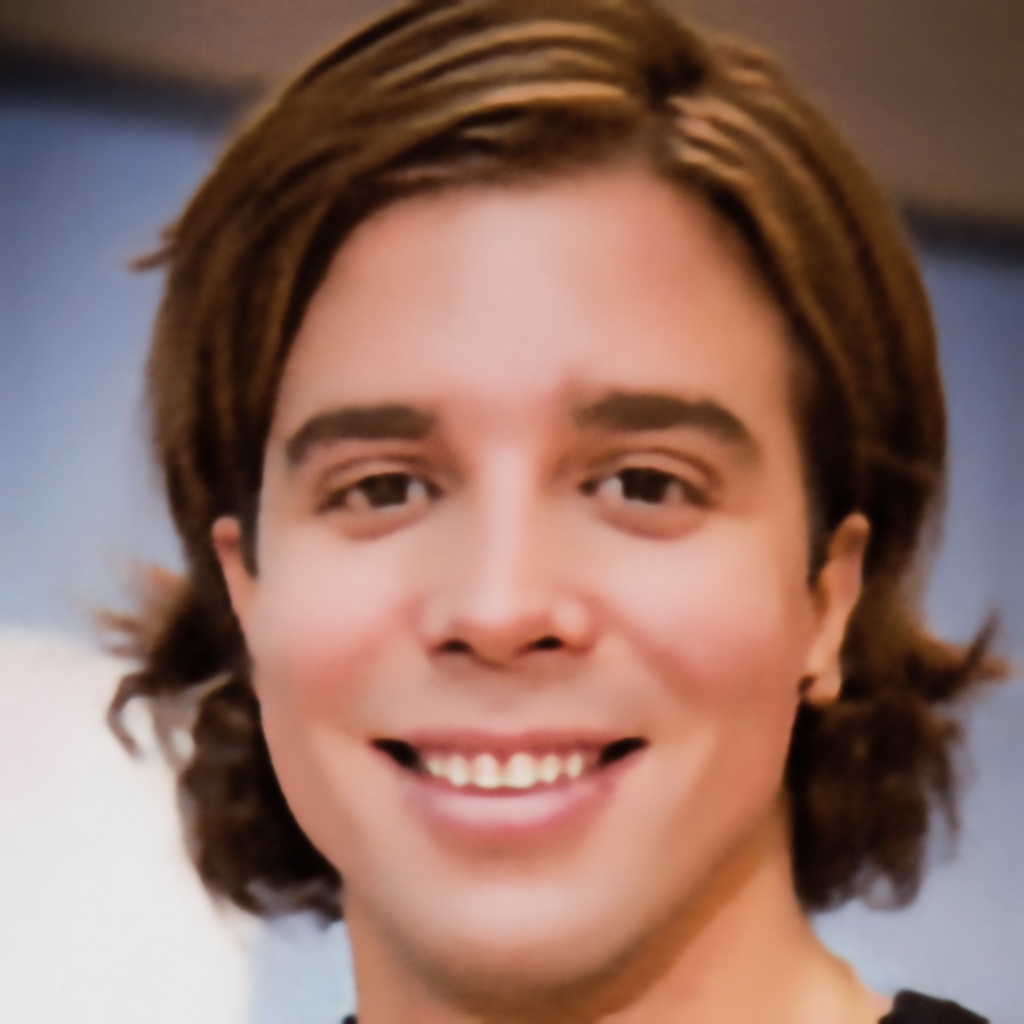}
  \end{subfigure}
  \begin{subfigure}[b]{0.32\linewidth}
    %\centering\includegraphics[width=75pt]{figs/imgHQ00075_CSGM_cs_0p5}
    \centering\includegraphics[width=75pt]{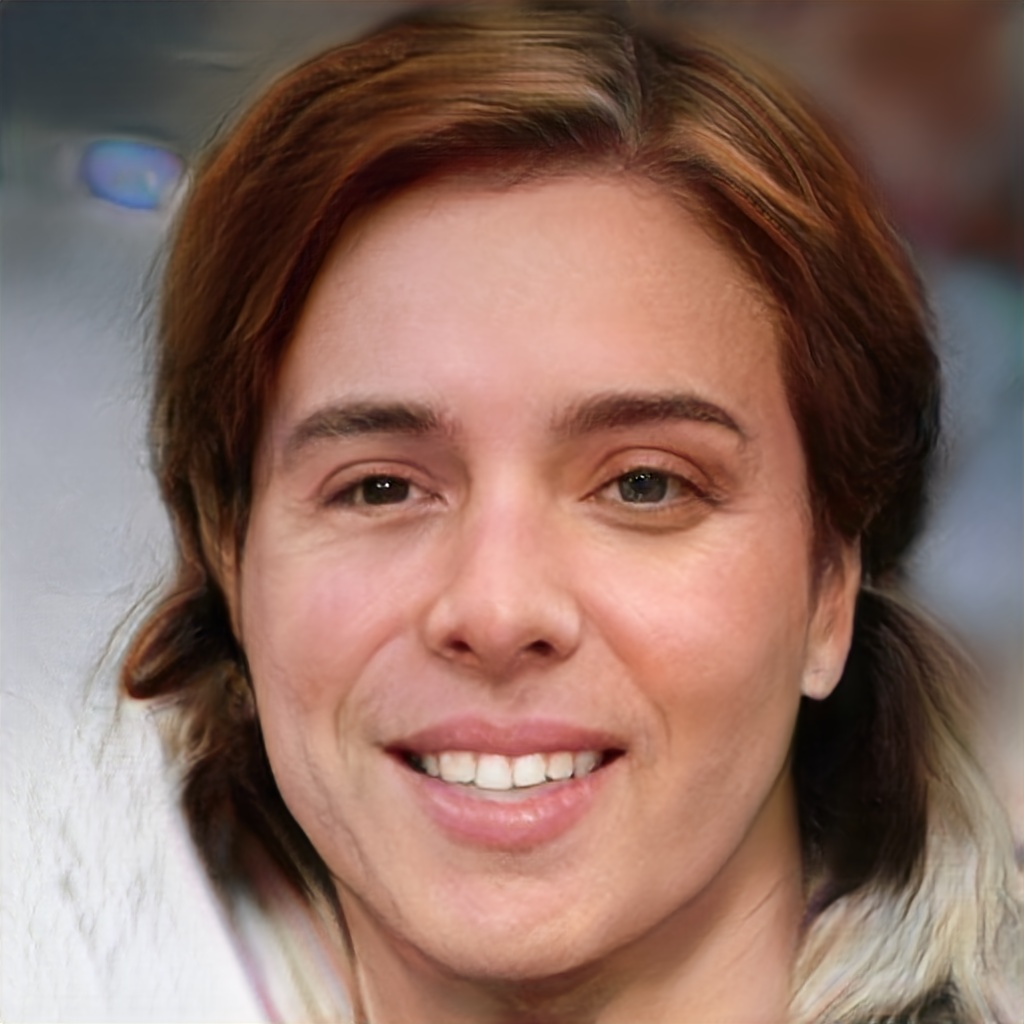}
  \end{subfigure}
  \begin{subfigure}[b]{0.32\linewidth}
    %\centering\includegraphics[width=75pt]{figs/imgHQ00075_IA_cs_0p5}
    \centering\includegraphics[width=75pt]{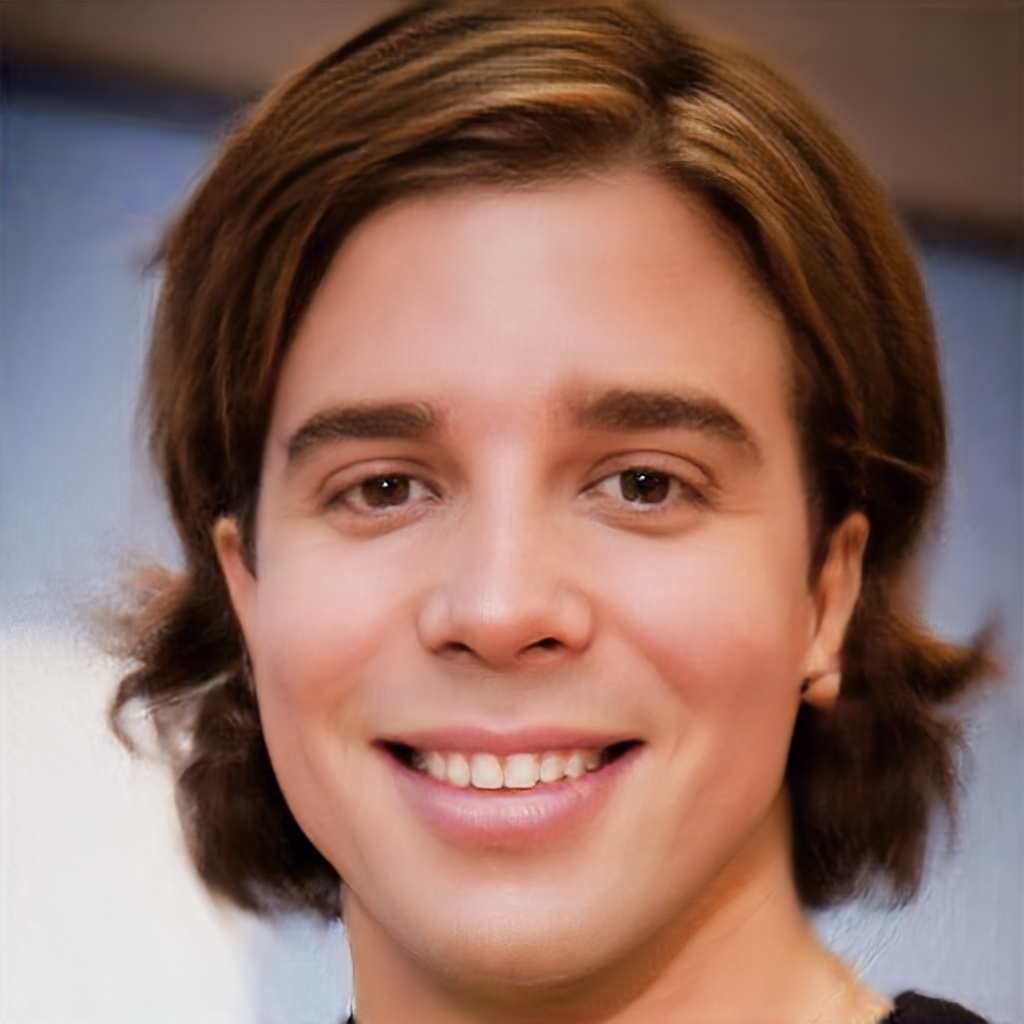}
  \end{subfigure}
  \caption{Deblurring with $9 \times 9$ uniform filter and noise level of 10/255, for CelebA-HQ images. From left to right and top to bottom (in each group): original image, blurred and noisy image, DIP, CSGM, and IAGAN. Note that CSGM and IAGAN use the PGGAN prior.}
\label{fig:deb_visual_hq}
\end{figure}

\section{Conclusion}

In this work we considered the usage of generative models for solving imaging inverse problems.
The main deficiency in such applications is the limited representation capabilities of the generators, which unfortunately do not capture the full distribution for complex classes of images.
We suggested two strategies for mitigating this problem. One technique is a post-processing back-projection step, which is applicable at low noise level, that essentially eliminates the component of the generator's representation error that resides in the row space of the measurement matrix. The second technique, which is our main contribution, is an image adaptive approach, termed IAGAN, that improves the generator capability to represent the {\em specific} test image. This method can improve also the restoration in the null space of the measurement matrix. One can also use the two strategies together.
Experiments on compressed sensing and super-resolution tasks demonstrated that our strategies, especially the image-adaptive approach, yield significantly improved reconstructions, which are both more accurate and perceptually pleasing than other alternatives.

%\newpage

\subsubsection*{Acknowledgments}

%SAH is supported by the NSF-BSF grant (No. 2017729). TT is supported by the European research council (ERC StG 757497 PI Giryes).

 This work is partially supported by the NSF-BSF grant (No. 2017729) and the European research council (ERC StG 757497 PI Giryes).

\bibliographystyle{aaai}
\bibliography{aaai20_TeX}

@book{goodfellow2016deep,
  title={Deep learning},
  author={Goodfellow, Ian and Bengio, Yoshua and Courville, Aaron},
  year={2016}
}

@article{kingma2013auto,
  title={Auto-encoding variational bayes},
  author={Kingma, Diederik P and Welling, Max},
  journal={arXiv preprint arXiv:1312.6114},
  year={2013}
}

@inproceedings{goodfellow2014generative,
  title={Generative adversarial nets},
  author={Goodfellow, Ian and Pouget-Abadie, Jean and Mirza, Mehdi and Xu, Bing and Warde-Farley, David and Ozair, Sherjil and Courville, Aaron and Bengio, Yoshua},
  booktitle={Advances in neural information processing systems},
  pages={2672--2680},
  year={2014}
}

@inproceedings{bojanowski2018optimizing,
  title={Optimizing the Latent Space of Generative Networks},
  author={Bojanowski, Piotr and Joulin, Armand and Lopez-Pas, David and Szlam, Arthur},
  booktitle={International Conference on Machine Learning},
  pages={599--608},
  year={2018}
}

@inproceedings{bora2017compressed,
  title={Compressed sensing using generative models},
  author={Bora, Ashish and Jalal, Ajil and Price, Eric and Dimakis, Alexandros G},
  booktitle={International Conference on Machine Learning},
  pages={537--546},
  year={2017},
}

@article{tirer2019image,
  title={Image restoration by iterative denoising and backward projections},
  author={Tirer, Tom and Giryes, Raja},
  journal={IEEE Transactions on Image Processing},
  volume={28},
  number={3},
  pages={1220--1234},
  year={2018},
  publisher={IEEE}
}

@article{tirer2018super,
  title={Super-Resolution via Image-Adapted Denoising {CNNs}: Incorporating External and Internal Learning},
  author={Tirer, Tom and Giryes, Raja},
  journal={IEEE Signal Processing Letters},
  year={2019},
  publisher={IEEE}
}

@inproceedings{liu2015deep,
  title={Deep learning face attributes in the wild},
  author={Liu, Ziwei and Luo, Ping and Wang, Xiaogang and Tang, Xiaoou},
  booktitle={Proceedings of the IEEE international conference on computer vision},
  pages={3730--3738},
  year={2015}
}

@article{kingma2014adam,
  title={Adam: A method for stochastic optimization},
  author={Kingma, Diederik P and Ba, Jimmy},
  journal={arXiv preprint arXiv:1412.6980},
  year={2014}
}

@inproceedings{hand2018phase,
  title={Phase retrieval under a generative prior},
  author={Hand, Paul and Leong, Oscar and Voroninski, Vlad},
  booktitle={Advances in Neural Information Processing Systems},
  pages={9154--9164},
  year={2018}
}

@inproceedings{shah2018solving,
  title={Solving linear inverse problems using gan priors: An algorithm with provable guarantees},
  author={Shah, Viraj and Hegde, Chinmay},
  booktitle={2018 IEEE International Conference on Acoustics, Speech and Signal Processing (ICASSP)},
  pages={4609--4613},
  year={2018},
  organization={IEEE}
}

@inproceedings{richardson2018gans,
  title={On {GANs} and {GMMs}},
  author={Richardson, Eitan and Weiss, Yair},
  booktitle={Advances in Neural Information Processing Systems},
  pages={5852--5863},
  year={2018}
}

@inproceedings{arjovsky2017wasserstein,
  title={Wasserstein generative adversarial networks},
  author={Arjovsky, Martin and Chintala, Soumith and Bottou, L{\'e}on},
  booktitle={International Conference on Machine Learning},
  pages={214--223},
  year={2017}
}

@article{karras2017progressive,
  title={Progressive growing of gans for improved quality, stability, and variation},
  author={Karras, Tero and Aila, Timo and Laine, Samuli and Lehtinen, Jaakko},
  journal={arXiv preprint arXiv:1710.10196},
  year={2017}
}

@inproceedings{dong2014learning,
  title={Learning a deep convolutional network for image super-resolution},
  author={Dong, Chao and Loy, Chen Change and He, Kaiming and Tang, Xiaoou},
  booktitle={European conference on computer vision},
  pages={184--199},
  year={2014},
  organization={Springer}
}

@inproceedings{meinhardt2017learning,
  title={Learning proximal operators: Using denoising networks for regularizing inverse imaging problems},
  author={Meinhardt, Tim and Moller, Michael and Hazirbas, Caner and Cremers, Daniel},
  booktitle={ICCV},
  pages={1781--1790},
  year={2017}
}

@inproceedings{rick2017one,
  title={One Network to Solve Them All--Solving Linear Inverse Problems Using Deep Projection Models},
  author={Rick Chang, JH and Li, Chun-Liang and Poczos, Barnabas and Vijaya Kumar, BVK and Sankaranarayanan, Aswin C},
  booktitle={ICCV},
  pages={5888--5897},
  year={2017}
}

@book{bertero1998introduction,
  title={Introduction to inverse problems in imaging},
  author={Bertero, Mario and Boccacci, Patrizia},
  year={1998},
  publisher={CRC press}
}

@inproceedings{venkatakrishnan2013plug,
  title={Plug-and-play priors for model based reconstruction},
  author={Venkatakrishnan, Singanallur V and Bouman, Charles A and Wohlberg, Brendt},
  booktitle={Global Conference on Signal and Information Processing (GlobalSIP), 2013 IEEE},
  pages={945--948},
  year={2013},
  organization={IEEE}
}

@InProceedings{ZSSR,
  author = {Shocher,Assaf and Cohen, Nadav and Irani, Michal},
  title = {"Zero-Shot" Super-Resolution using Deep Internal Learning},
  booktitle={CVPR},
  month = {June},
  year = {2018}
}

@inproceedings{ulyanov2017deep,
  title={Deep image prior},
  author={Ulyanov, Dmitry and Vedaldi, Andrea and Lempitsky, Victor},
  booktitle={CVPR},
  year={2018}
}

@inproceedings{zhang2017learning,
   title={Learning Deep CNN Denoiser Prior for Image Restoration},
   author={Zhang, Kai and Zuo, Wangmeng and Gu, Shuhang and Zhang, Lei},
   booktitle={IEEE Conference on Computer Vision and Pattern Recognition},
   pages={3929--3938},
   year={2017},
 }

@inproceedings{yeh2017semantic,
  title={Semantic image inpainting with deep generative models},
  author={Yeh, Raymond A and Chen, Chen and Yian Lim, Teck and Schwing, Alexander G and Hasegawa-Johnson, Mark and Do, Minh N},
  booktitle={Proceedings of the IEEE Conference on Computer Vision and Pattern Recognition},
  pages={5485--5493},
  year={2017}
}

@inproceedings{glasner2009super,
  title={Super-resolution from a single image},
  author={Glasner, Daniel and Bagon, Shai and Irani, Michal},
  booktitle={Computer Vision, 2009 IEEE 12th International Conference on},
  pages={349--356},
  year={2009},
  organization={IEEE}
}

@article{yang2010image,
  title={Image super-resolution via sparse representation},
  author={Yang, Jianchao and Wright, John and Huang, Thomas S and Ma, Yi},
  journal={IEEE transactions on image processing},
  volume={19},
  number={11},
  pages={2861--2873},
  year={2010},
  publisher={IEEE}
}

@article{berthelot2017began,
  title={Began: Boundary equilibrium generative adversarial networks},
  author={Berthelot, David and Schumm, Thomas and Metz, Luke},
  journal={arXiv preprint arXiv:1703.10717},
  year={2017}
}

@book{hestenes1952methods,
  title={Methods of conjugate gradients for solving linear systems},
  author={Hestenes, Magnus Rudolph and Stiefel, Eduard},
  volume={49},
  number={1},
  year={1952}
}

@article{lecun1998gradient,
  title={Gradient-based learning applied to document recognition},
  author={LeCun, Yann and Bottou, L{\'e}on and Bengio, Yoshua and Haffner, Patrick and others},
  journal={Proceedings of the IEEE},
  volume={86},
  number={11},
  pages={2278--2324},
  year={1998},
  publisher={Taipei, Taiwan}
}

@inproceedings{bora2018ambientgan,
  title={AmbientGAN: Generative models from lossy measurements},
  author={Bora, Ashish and Price, Eric and Dimakis, Alexandros G},
  booktitle={International Conference on Learning Representations (ICLR)},
  year={2018}
}

@article{lustig2007sparse,
  title={Sparse {MRI}: The application of compressed sensing for rapid MR imaging},
  author={Lustig, Michael and Donoho, David and Pauly, John M},
  journal={Magnetic Resonance in Medicine: An Official Journal of the International Society for Magnetic Resonance in Medicine},
  volume={58},
  number={6},
  pages={1182--1195},
  year={2007},
  publisher={Wiley Online Library}
}

@inproceedings{blau2018perception,
  title={The perception-distortion tradeoff},
  author={Blau, Yochai and Michaeli, Tomer},
  booktitle={Proceedings of the IEEE Conference on Computer Vision and Pattern Recognition},
  pages={6228--6237},
  year={2018}
}

@article{donoho2006compressed,
  title={Compressed sensing},
  author={Donoho, David},
  journal={IEEE Transactions on information theory},
  volume={52},
  number={4},
  pages={1289--1306},
  year={2006},
  publisher={Citeseer}
}

@article{candes2004robust,
  title={Robust uncertainty principles: Exact signal reconstruction from highly incomplete frequency information},
  author={Candes, Emmanuel and Romberg, Justin and Tao, Terence},
  journal={IEEE Transactions on information theory},
  volume={52},
  number={2},
  pages={489--509},
  year={2006}
}

@article{metzler2016denoising,
  title={From denoising to compressed sensing},
  author={Metzler, Christopher A and Maleki, Arian and Baraniuk, Richard G},
  journal={IEEE Transactions on Information Theory},
  volume={62},
  number={9},
  pages={5117--5144},
  year={2016},
  publisher={IEEE}
}

@inproceedings{zhang2018unreasonable,
  title={The unreasonable effectiveness of deep features as a perceptual metric},
  author={Zhang, Richard and Isola, Phillip and Efros, Alexei A and Shechtman, Eli and Wang, Oliver},
  booktitle={Proceedings of the IEEE Conference on Computer Vision and Pattern Recognition},
  pages={586--595},
  year={2018}
}

@article{dhar2018modeling,
  title={Modeling sparse deviations for compressed sensing using generative models},
  author={Dhar, Manik and Grover, Aditya and Ermon, Stefano},
  journal={arXiv preprint arXiv:1807.01442},
  year={2018}
}

\end{document}